\documentclass[lettersize,draftcls,onecolumn]{IEEEtran}
\usepackage{algorithmic}
\usepackage{algorithm}
\usepackage{array}
\usepackage[caption=false,font=normalsize,labelfont=sf,textfont=sf]{subfig}
\usepackage{textcomp}
\usepackage{stfloats}
\usepackage{url}
\usepackage{verbatim}
\usepackage{graphicx}
\usepackage{cite}
\usepackage{amsmath, amsthm, amssymb, amsfonts}
\usepackage{thmtools}
\usepackage{graphicx}
\usepackage{setspace}
\usepackage{float}
\usepackage{hyperref}
\usepackage[utf8]{inputenc}
\usepackage[english]{babel}
\usepackage{framed}
\usepackage[dvipsnames]{xcolor}
\usepackage{tcolorbox}
\usepackage{tikz}
\usepackage{etoolbox}
\usepackage{amsthm}
\usepackage{xcolor}
\usepackage{calc}
\usepackage{svg}
\usepackage{svg-extract}
\svgsetup{clean}
\usepackage{chngcntr}
\usepackage{enumitem}
\setlist{nosep}
\usepackage{graphicx}
\usepackage{caption}
\usepackage{subcaption}
\usepackage{subfig}
\usepackage{float}  
\usepackage{stfloats}  
\newtheorem{thm}{Theorem} 
\newtheorem{lem}{Lemma}
\newtheorem{cor}{Corollary}
\newtheorem{defn}{Definition}

\newcommand{\Trans}{\operatorname{Trans}}

\begin{document}

\title{Distributed Batch Matrix Multiplication: Trade-Offs in Download Rate, Randomness, and Privacy}

\author{Amirhosein~Morteza and
        R\'emi~A.~Chou
\thanks{The authors are with the Department of Computer Science and Engineering, The University of Texas at Arlington, Arlington, TX 76019. This work was supported in part by NSF grant CCF-2425371. Part of this work was presented at the  60th Annual Allerton Conference on Communication, Control, and Computing, Allerton’24 \cite{allerton2024}. 
}}

\markboth{}%
{Morteza \MakeLowercase{\textit{et al.}}: Distributed Matrix Multiplication: Download Rate, Randomness and Privacy Trade-Offs}


\maketitle

\begin{abstract}
We study the trade-off between communication rate and privacy for distributed batch matrix multiplication of two independent sequences of matrices $\mathbf{A}$ and $\mathbf{B}$ with uniformly distributed entries. In our setting, $\mathbf{B}$ is publicly accessible by all the servers while $\mathbf{A}$ must remain private. A user is interested in evaluating the product $\mathbf{AB}$ with the responses from the $k$ fastest servers. For a given parameter $\alpha \in [0, 1]$, our privacy constraint must ensure that any set of $\ell$ colluding servers cannot learn more than a fraction $\alpha$ of $\mathbf{A}$. Additionally, we study the trade-off between the amount of local randomness needed at the encoder and privacy. Finally, we establish the optimal trade-offs when the matrices are square and identify a linear relationship between information leakage and communication rate.
\end{abstract}


\section{Introduction}
\IEEEPARstart{T}{he} task of multiparty computation with security guarantees, first explored in \cite{yao1982protocols,chaum1988multiparty}, has recently expanded to include outsourcing large-scale matrix multiplication tasks to distributed servers to speed up computation. It has applications in machine learning, signal processing, data encryption, and computational efficiency in cloud computing, e.g., \cite{soleymani2020privacy,fu2017secure,zhu2020secure,kakar2019capacity}.  For instance, privacy-preserving machine learning may involve training on encrypted data and protecting private information like health records.    In several applications, only one matrix needs to be private.  \cite{jiang2018secure} and \cite{gilad2016cryptonets} explore secure outsourcing matrix computations in neural networks, requiring the privacy of input matrices ($\mathbf{A}$) while allowing computation results or model parameters ($\mathbf{B}$) to be openly handled. This setup exemplifies the principle of keeping sensitive data private while utilizing public model parameters.

The information-theoretic investigation of secure distributed matrix multiplication emerged in \cite{chang2018capacity}, where two matrices $A$ and $B$ are securely encoded and transmitted to $N$ servers by the user, who retrieves $AB$ from the data downloaded. However, a scenario where a controlled amount of information leakage is permissible can help reduce communication complexity. To study the trade-off between privacy leakage and communication rate, we consider a setting with a newly defined privacy constraint that allows a controlled amount of leakage and design a coding scheme that meets such constraint.

In this paper, we study the problem of secure batch matrix multiplication for two sequences of matrices, $\mathbf{A}$ and $\mathbf{B}$, independently and uniformly distributed over a finite field. The user is interested in distributing the computation of the product $\mathbf{AB}$ over $N$ servers. By downloading responses from the $k$ fastest servers, the user can retrieve $\mathbf{AB}$ (Recoverability Constraint). The download rate is the ratio of the number of bits required to represent the computation result to the total number of bits that the servers must transmit to the user. Unlike previous studies, which considered perfect privacy, e.g., \cite{jia2021capacity,kakar2019capacity,chang2018capacity}, in our setting, a controlled amount of information leakage is permissible, meaning that, for a given parameter $\alpha \in [0, 1]$, no more than a fraction $\alpha$ of information about $\mathbf{A}$ can be learned by any set of $\ell$ colluding servers (Privacy Constraint). The capacity is defined as the supremum of the download rate. The capacity of this model has been characterized in \cite{chang2018capacity, jia2021capacity} when $\alpha = 0$, and its characterization when the matrices do not have uniformly distributed entries is an open problem \cite{jia2021capacity}. Our main contributions are:
 \begin{itemize}
    \item [(i)] Formalizing a new problem setting that enables the study of the trade-off between download rate and privacy leakage. Our results generalize \cite[Theorem 1]{chang2018capacity} and \cite[Theorem 3]{jia2021capacity} obtained when $\alpha = 0$.  
    \item [(ii)] Understanding the trade-off between privacy leakage and local randomness to efficiently use this costly resource at the encoder. Specifically, we determine bounds for the amount of randomness needed to meet the privacy constraint, which to the best of our knowledge, no previous works had investigated, even when $\alpha = 0$. 
    \item [(iii)] Characterizing the capacity for square matrices and showing a download rate gain of \(\min \left(\frac{(k - \ell) \alpha}{k(1 - \alpha)}, \frac{\ell}{k}\right)\) compared to the case $\alpha  = 0$. We also identify a linear relationship between communication rate and privacy leakage by finding that the capacity is proportional to \(\frac{1}{1- \alpha}\) when $\alpha < \ell / k$. Finally, we establish the optimal rate of local randomness needed at the encoder and show a gain of $\min \left(\frac{\alpha k}{k - \ell}, \frac{\ell}{k - \ell}\right)$ compared to the case $\alpha = 0$.
\end{itemize}

\subsection{Related works} \label{related_works}

The information-theoretic investigation of secure distributed matrix multiplication emerged in \cite{chang2018capacity}, where two matrices $A$ and $B$ are securely encoded and transmitted to $N$ servers by the user who retrieves $AB$ from the data downloaded. The authors designed two models: a one-sided secure model where only matrix $A$ is private, and a fully secured model where both matrices are private. Recent works, e.g., \cite{zhu2020secure, aliasgari2020private, kakar2019capacity}, aimed to optimize communication overheads. Recently, \cite{jia2021capacity} considered the batch multiplication of a sequence of matrices with uniformly distributed entries and showed that the capacity of a multi-message X-secure T-private information retrieval problem provides an upper bound on the download rate. This problem is also explored in \cite{lopez2022secure, d2020gasp,yu2020coded,yu2020entangled, makkonen2023algebraic,machado2023hera,das2021efficient,chen2021gcsa,hasirciouglu2021bivariate,pradhan2021factored} to investigate the use of coding techniques to reduce communication costs. Additionally, \cite{mital2022secure} and \cite{nodehi2018limited} considered the setting with distributed nodes where data does not originate at the user requesting the computation. These works focus on reducing the download and upload rate while preserving privacy.  

Note that the above references only considered perfect security without any information leakage in their problem settings. Notably, \cite{bitar2024sparsity} has explored the trade-off between privacy and sparsity in distributed computing by examining sparse secret-sharing schemes, where increased sparsity results
in weaker privacy. We also note that related works have investigated the general research direction considered in this paper, namely communication complexity of secure function computation under information leakage, e.g., \cite{chou2022function,chou2024private}.

Compared to previous works, our study demonstrates how relaxing privacy constraints can enhance communication efficiency in distributed matrix multiplication.

\subsection{Main differences with previous works}

 Previous works consider no privacy leakage, ensuring that any set of colluding servers cannot obtain any information about private matrices. To the best of our knowledge, no previous work studies the capacity for secure distributed matrix multiplication when information leakage is allowed. In this study, we define a new problem setting that incorporates a privacy constraint with an information leakage parameter $\alpha$, which could not be addressed with previous techniques, as detailed next. With such a privacy constraint, we obtain bounds on the capacity that generalize previously found bounds in \cite[Theorem 1]{chang2018capacity} and \cite[Theorem 3]{jia2021capacity}. Another contrast between this work and previous studies is that we bound the optimal rate of local randomness needed at the encoder to satisfy the privacy constraint, making this the first study to explore such bounds, even when $\alpha =0$.

In our setting, a significant challenge lies in satisfying the new privacy constraint, preventing any set of $\ell < k$ colluding servers from learning more than a fraction \(\alpha \in [0, 1]\) of information about \(\mathbf{A}\). We cannot rely on traditional secret sharing \cite{ adi1979share}, which does not allow any information leakage.  Inspired by recent studies on the trade-offs among privacy, communication rate, and storage requirement \cite{chou2024secure,Maryam-Remi2023,morteza,chou2020secure} in the context of secret sharing, we leverage the combination of two ramp secret-sharing schemes \cite{yamamoto1986secret} —a strategy not previously investigated— to perform matrix multiplication and allow a controlled privacy leakage. 
Specifically, our proposed construction distinguishes two cases: 
(i) $\alpha \geq \frac{\ell}{k}$, where we consider a ramp secret-sharing scheme where no local randomness is needed, and 
(ii) $\alpha < \frac{\ell}{k}$, where local randomness is necessary to ensure privacy, and we use a combination of two ramp secret-sharing schemes.

 Another distinction in this study lies in the analysis of the proposed coding scheme. We require a more precise analysis to accommodate the combination of two ramp secret-sharing schemes, as the analysis of a single ramp secret-sharing scheme, such as in \cite{jia2021capacity}, is not sufficient for our case. Compared to previous studies, which only considered the case \(\alpha = 0\), where no information is leaked, our analysis reveals a trade-off among communication rate, local randomness rate, and privacy.

\subsection{Paper organization}
The remainder of the paper is organized as follows. We define the problem in Section~\ref{Problem Statement} and present our main results in Section~\ref{Main Results}. We discuss our converse and achievability in Sections~\ref{Converse} and \ref{Achievability}, respectively. Finally, we provide concluding remarks in Section~\ref{Conclusion}.

\section{Problem Statement} \label{Problem Statement}

\textbf{Notation}: Let $\mathbb{F}_{q}$ be a finite field characterized by a large prime number $q$. Let $\mathbb{Q}$, $\mathbb{N}$, and $\mathbb{R}$ be the set of rational, natural, and real numbers, respectively. For any $a, b \in \mathbb{R}$, define $[a] \triangleq [1, \lceil a \rceil] \cap \mathbb{N}$, $[a: b] \triangleq [\lfloor a \rfloor , \lceil b \rceil] \cap \mathbb{N}$, and $[a: b) \triangleq [a:b-1]$. Sets are represented by calligraphic letters, and sequences of matrices are represented by bold uppercase letters. Let $[a]^{= b} \overset{\mathrm{\Delta}}{=} \{\mathcal{I} \subseteq [a]: |\mathcal{I}| = b \}$ be the set of all the subsets of $[a]$ that have cardinality $b$; $[a]^{\leq b} \overset{\mathrm{\Delta}}{=} \{\mathcal{I} \subseteq [a]: |\mathcal{I}| \leq b \}$ be the set of all the subsets of $[a]$ that have a cardinality less than or equal to $b$. Logarithms are defined with base $q$. Also, define $[a]^{+} \triangleq \max \{0, a\}$. The transpose of a matrix $M$ is denoted by $\Trans(M)$.

\begin{defn}\label{Def_1}
 Let $N, r\in \mathbb{N}$, and $k \in [N]$. An $(N, k, r)$-coding scheme consists of 
\begin{itemize}
    \item $N\geq 2$ servers; 
    \item Two sequences of independent matrices, $\mathbf{A} \triangleq (A_{s})_{s \in [m]}$ and $\mathbf{B} \triangleq (B_{s})_{s \in [m]}$, where $m$ is a large integer. For any $s \in [m]$, the matrices $A_s$ and $B_s$ are assumed to be uniformly distributed over $\mathbb{F}_q^{C \times D}$ and $\mathbb{F}_q^{D \times E}$, respectively. $\mathbf{B}$ is public, while $\mathbf{A}$ is private and accessible only by the user;
    \item Local randomness in the form of a uniform random variable $R$ which is distributed over $\mathbb{F}_{q}^{r}$ and  independent of $(\mathbf{A},\mathbf{B})$;
    \item $N$ encoding functions $f_{i}: (\mathbf{A}, R) \mapsto \tilde{\mathbf{A}}_{i}$, $i \in [N]$, such that $\mathbf{A}$ can be recovered from encoded matrices $\tilde{\mathbf{A}}_{\mathcal{I}} \triangleq (\tilde{\mathbf{A}}_{i})_{i\in \mathcal{I}}$, $\mathcal{I} \in [N]^{\geq k}$, i.e.,
     \begin{align}
        H(\mathbf{A}|\tilde{\mathbf{A}}_{\mathcal{I}}) = 0; \label{recover_all}    
    \end{align}
    \item $N$ processing functions $ h_{i}: (\tilde{\mathbf{A}}_{i}, \mathbf{B}) \mapsto Z_{i}$, $i \in [N]$;
    \item A decoding function $d$ taking $Z_{\mathcal{I}} \triangleq (Z_i)_{i \in \mathcal{I}}$, $\mathcal{I} \subseteq [N]$, and returning an estimate of $\mathbf{AB} \triangleq (A_{s}B_{s})_{s \in [m]}$;
\end{itemize}
and operates as follows:
\begin{itemize}
    \item For all $i \in [N]$, the user sends the encoded matrices $\tilde{\mathbf{A}}_{i} \triangleq f_i(\mathbf{A}, R)$ to Server~$i$ over a private channel;
    \item For all $i \in [N]$, Server~$i$ generates a response $Z_i \triangleq h_i(\tilde{\mathbf{A}}_{i}, \mathbf{B})$;
    \item The user computes an estimate of $\mathbf{AB}$ as $d(Z_{\mathcal{I}})$, where $Z_{\mathcal{I}}$ is a sequence of received responses from the $k$ fastest~servers. 
\end{itemize}

\end{defn}

\begin{figure}[htbp] \label{Fig_1}

\centering
\tikzset {_4iaqce4wz/.code = {\pgfsetadditionalshadetransform{ \pgftransformshift{\pgfpoint{0 bp } { 0 bp }  }  \pgftransformrotate{0 }  \pgftransformscale{2 }  }}}
\pgfdeclarehorizontalshading{_3spspmajk}{150bp}{rgb(0bp)=(1,1,1);
rgb(62.5bp)=(1,1,1);
rgb(62.5bp)=(0.95,0.95,0.95);
rgb(62.5bp)=(0.93,0.93,0.93);
rgb(62.5bp)=(1,1,1);
rgb(100bp)=(1,1,1)}

  
\tikzset {_afvizhn5r/.code = {\pgfsetadditionalshadetransform{ \pgftransformshift{\pgfpoint{0 bp } { 0 bp }  }  \pgftransformrotate{0 }  \pgftransformscale{2 }  }}}
\pgfdeclarehorizontalshading{_gk7dla613}{150bp}{rgb(0bp)=(1,1,1);
rgb(37.5bp)=(1,1,1);
rgb(62.291665758405415bp)=(1,1,1);
rgb(62.5bp)=(0.95,0.95,0.95);
rgb(62.5bp)=(0.93,0.93,0.93);
rgb(100bp)=(0.93,0.93,0.93)}

  
\tikzset {_kmaaiezm6/.code = {\pgfsetadditionalshadetransform{ \pgftransformshift{\pgfpoint{0 bp } { 0 bp }  }  \pgftransformrotate{0 }  \pgftransformscale{2 }  }}}
\pgfdeclarehorizontalshading{_3ynrhjudw}{150bp}{rgb(0bp)=(1,1,1);
rgb(62.5bp)=(1,1,1);
rgb(62.5bp)=(0.95,0.95,0.95);
rgb(62.5bp)=(0.93,0.93,0.93);
rgb(62.5bp)=(1,1,1);
rgb(100bp)=(1,1,1)}

  
\tikzset {_o0ndvtxed/.code = {\pgfsetadditionalshadetransform{ \pgftransformshift{\pgfpoint{0 bp } { 0 bp }  }  \pgftransformrotate{0 }  \pgftransformscale{2 }  }}}
\pgfdeclarehorizontalshading{_8oqefxvho}{150bp}{rgb(0bp)=(1,1,1);
rgb(62.5bp)=(1,1,1);
rgb(62.5bp)=(0.9,0.9,0.9);
rgb(100bp)=(0.9,0.9,0.9)}

  
\tikzset {_6uh338bvg/.code = {\pgfsetadditionalshadetransform{ \pgftransformshift{\pgfpoint{0 bp } { 0 bp }  }  \pgftransformrotate{0 }  \pgftransformscale{2 }  }}}
\pgfdeclarehorizontalshading{_jxkzcmql9}{150bp}{rgb(0bp)=(1,1,1);
rgb(62.5bp)=(1,1,1);
rgb(62.5bp)=(0.9,0.9,0.9);
rgb(100bp)=(0.9,0.9,0.9)}

  
\tikzset {_e2xg8l0yr/.code = {\pgfsetadditionalshadetransform{ \pgftransformshift{\pgfpoint{0 bp } { 0 bp }  }  \pgftransformrotate{0 }  \pgftransformscale{2 }  }}}
\pgfdeclarehorizontalshading{_jl11j7z0x}{150bp}{rgb(0bp)=(1,1,1);
rgb(62.5bp)=(1,1,1);
rgb(62.5bp)=(0.9,0.9,0.9);
rgb(100bp)=(0.9,0.9,0.9)}

  
\tikzset {_eq40yeb6f/.code = {\pgfsetadditionalshadetransform{ \pgftransformshift{\pgfpoint{0 bp } { 0 bp }  }  \pgftransformrotate{0 }  \pgftransformscale{2 }  }}}
\pgfdeclarehorizontalshading{_meb3eig1o}{150bp}{rgb(0bp)=(1,1,1);
rgb(62.5bp)=(1,1,1);
rgb(62.5bp)=(0.9,0.9,0.9);
rgb(100bp)=(0.9,0.9,0.9)}

  
\tikzset {_idtn18m1f/.code = {\pgfsetadditionalshadetransform{ \pgftransformshift{\pgfpoint{0 bp } { 0 bp }  }  \pgftransformrotate{0 }  \pgftransformscale{2 }  }}}
\pgfdeclarehorizontalshading{_tsi6g1c3q}{150bp}{rgb(0bp)=(1,1,1);
rgb(62.5bp)=(1,1,1);
rgb(62.5bp)=(0.9,0.9,0.9);
rgb(100bp)=(0.9,0.9,0.9)}

  
\tikzset {_acwmx1e7t/.code = {\pgfsetadditionalshadetransform{ \pgftransformshift{\pgfpoint{0 bp } { 0 bp }  }  \pgftransformrotate{0 }  \pgftransformscale{2 }  }}}
\pgfdeclarehorizontalshading{_78v7e9k6w}{150bp}{rgb(0bp)=(1,1,1);
rgb(62.5bp)=(1,1,1);
rgb(62.5bp)=(0.9,0.9,0.9);
rgb(100bp)=(0.9,0.9,0.9)}
\tikzset{every picture/.style={line width=0.75pt, font=\large}} 
\scalebox{0.6}{
\begin{tikzpicture}[x=0.75pt,y=0.75pt,yscale=-1,xscale=1]

\draw  [color={rgb, 255:red, 10; green, 59; blue, 112 }  ,draw opacity=1 ] (75,261.16) .. controls (75,248.37) and (93.33,238) .. (115.94,238) .. controls (138.55,238) and (156.88,248.37) .. (156.88,261.16) .. controls (156.88,273.96) and (138.55,284.33) .. (115.94,284.33) .. controls (93.33,284.33) and (75,273.96) .. (75,261.16) -- cycle ;
\draw  [color={rgb, 255:red, 119; green, 108; blue, 4 }  ,draw opacity=1 ] (301,440.45) .. controls (301,424.74) and (313.74,412) .. (329.45,412) .. controls (345.16,412) and (357.89,424.74) .. (357.89,440.45) .. controls (357.89,456.16) and (345.16,468.89) .. (329.45,468.89) .. controls (313.74,468.89) and (301,456.16) .. (301,440.45) -- cycle ;
\draw [color={rgb, 255:red, 208; green, 2; blue, 27 }  ,draw opacity=1 ]   (149.33,69.67) -- (111.54,235.51) ;
\draw [shift={(110.88,238.44)}, rotate = 282.84] [fill={rgb, 255:red, 208; green, 2; blue, 27 }  ,fill opacity=1 ][line width=0.08]  [draw opacity=0] (8.93,-4.29) -- (0,0) -- (8.93,4.29) -- cycle    ;
\draw [color={rgb, 255:red, 208; green, 2; blue, 27 }  ,draw opacity=1 ]   (149.33,69.67) -- (300.37,244.4) ;
\draw [shift={(302.33,246.67)}, rotate = 229.16] [fill={rgb, 255:red, 208; green, 2; blue, 27 }  ,fill opacity=1 ][line width=0.08]  [draw opacity=0] (8.93,-4.29) -- (0,0) -- (8.93,4.29) -- cycle    ;
\draw [color={rgb, 255:red, 126; green, 211; blue, 33 }  ,draw opacity=1 ]   (549.88,70.33) -- (159.58,259.85) ;
\draw [shift={(156.88,261.16)}, rotate = 334.1] [fill={rgb, 255:red, 126; green, 211; blue, 33 }  ,fill opacity=1 ][line width=0.08]  [draw opacity=0] (8.93,-4.29) -- (0,0) -- (8.93,4.29) -- cycle    ;
\draw [color={rgb, 255:red, 126; green, 211; blue, 33 }  ,draw opacity=1 ]   (549.88,70.33) -- (354.12,244.44) ;
\draw [shift={(351.88,246.44)}, rotate = 318.35] [fill={rgb, 255:red, 126; green, 211; blue, 33 }  ,fill opacity=1 ][line width=0.08]  [draw opacity=0] (8.93,-4.29) -- (0,0) -- (8.93,4.29) -- cycle    ;
\draw [color={rgb, 255:red, 126; green, 211; blue, 33 }  ,draw opacity=1 ]   (549.88,70.33) -- (568.49,241.18) ;
\draw [shift={(568.82,244.16)}, rotate = 263.78] [fill={rgb, 255:red, 126; green, 211; blue, 33 }  ,fill opacity=1 ][line width=0.08]  [draw opacity=0] (8.93,-4.29) -- (0,0) -- (8.93,4.29) -- cycle    ;
\draw [color={rgb, 255:red, 208; green, 2; blue, 27 }  ,draw opacity=1 ]   (149.33,69.67) -- (529.19,257.11) ;
\draw [shift={(531.88,258.44)}, rotate = 206.26] [fill={rgb, 255:red, 208; green, 2; blue, 27 }  ,fill opacity=1 ][line width=0.08]  [draw opacity=0] (8.93,-4.29) -- (0,0) -- (8.93,4.29) -- cycle    ;
\draw [color={rgb, 255:red, 12; green, 77; blue, 159 }  ,draw opacity=1 ][line width=1.5]  [dash pattern={on 1.69pt off 2.76pt}]  (568.82,290.49) -- (357.3,420.05) ;
\draw [shift={(353.89,422.14)}, rotate = 328.51] [fill={rgb, 255:red, 12; green, 77; blue, 159 }  ,fill opacity=1 ][line width=0.08]  [draw opacity=0] (11.61,-5.58) -- (0,0) -- (11.61,5.58) -- cycle    ;
\draw  [color={rgb, 255:red, 208; green, 2; blue, 27 }  ,draw opacity=1 ] (79,29) -- (245,29) -- (245,69) -- (79,69) -- cycle ;
\draw  [color={rgb, 255:red, 65; green, 117; blue, 5 }  ,draw opacity=1 ] (464,30) -- (626,30) -- (626,70) -- (464,70) -- cycle ;
\draw  [color={rgb, 255:red, 10; green, 59; blue, 112 }  ,draw opacity=1 ] (288,265.16) .. controls (288,252.37) and (306.33,242) .. (328.94,242) .. controls (351.55,242) and (369.88,252.37) .. (369.88,265.16) .. controls (369.88,277.96) and (351.55,288.33) .. (328.94,288.33) .. controls (306.33,288.33) and (288,277.96) .. (288,265.16) -- cycle ;
\draw  [color={rgb, 255:red, 10; green, 59; blue, 112 }  ,draw opacity=1 ] (527.88,267.33) .. controls (527.88,254.53) and (546.21,244.16) .. (568.82,244.16) .. controls (591.43,244.16) and (609.76,254.53) .. (609.76,267.33) .. controls (609.76,280.12) and (591.43,290.49) .. (568.82,290.49) .. controls (546.21,290.49) and (527.88,280.12) .. (527.88,267.33) -- cycle ;
\draw [color={rgb, 255:red, 12; green, 77; blue, 159 }  ,draw opacity=1 ][line width=1.5]  [dash pattern={on 1.69pt off 2.76pt}]  (126.88,283.44) -- (300.77,422.64) ;
\draw [shift={(303.89,425.14)}, rotate = 218.68] [fill={rgb, 255:red, 12; green, 77; blue, 159 }  ,fill opacity=1 ][line width=0.08]  [draw opacity=0] (11.61,-5.58) -- (0,0) -- (11.61,5.58) -- cycle    ;
\draw [color={rgb, 255:red, 12; green, 77; blue, 159 }  ,draw opacity=1 ][line width=1.5]  [dash pattern={on 1.69pt off 2.76pt}]  (328.94,288.33) -- (329.43,408) ;
\draw [shift={(329.45,412)}, rotate = 269.77] [fill={rgb, 255:red, 12; green, 77; blue, 159 }  ,fill opacity=1 ][line width=0.08]  [draw opacity=0] (11.61,-5.58) -- (0,0) -- (11.61,5.58) -- cycle    ;
\path  [shading=_3spspmajk,_4iaqce4wz] (454.3,323.41) -- (491.4,323.41) -- (491.4,360.51) -- (454.3,360.51) -- cycle ; 
 \draw  [color={rgb, 255:red, 24; green, 90; blue, 163 }  ,draw opacity=1 ] (454.3,323.41) -- (491.4,323.41) -- (491.4,360.51) -- (454.3,360.51) -- cycle ; 

\path  [shading=_gk7dla613,_afvizhn5r] (191.86,326.33) -- (228.96,326.33) -- (228.96,363.43) -- (191.86,363.43) -- cycle ; 
 \draw  [color={rgb, 255:red, 24; green, 90; blue, 163 }  ,draw opacity=1 ] (191.86,326.33) -- (228.96,326.33) -- (228.96,363.43) -- (191.86,363.43) -- cycle ; 

\path  [shading=_3ynrhjudw,_kmaaiezm6] (312.78,324.78) -- (349.88,324.78) -- (349.88,361.88) -- (312.78,361.88) -- cycle ; 
 \draw  [color={rgb, 255:red, 24; green, 90; blue, 163 }  ,draw opacity=1 ] (312.78,324.78) -- (349.88,324.78) -- (349.88,361.88) -- (312.78,361.88) -- cycle ; 

\path  [shading=_8oqefxvho,_o0ndvtxed] (112.78,129.78) -- (149.88,129.78) -- (149.88,166.88) -- (112.78,166.88) -- cycle ; 
 \draw  [color={rgb, 255:red, 208; green, 2; blue, 27 }  ,draw opacity=1 ] (112.78,129.78) -- (149.88,129.78) -- (149.88,166.88) -- (112.78,166.88) -- cycle ; 

\path  [shading=_jxkzcmql9,_6uh338bvg] (202.78,128.78) -- (239.88,128.78) -- (239.88,165.88) -- (202.78,165.88) -- cycle ; 
 \draw  [color={rgb, 255:red, 208; green, 2; blue, 27 }  ,draw opacity=1 ] (202.78,128.78) -- (239.88,128.78) -- (239.88,165.88) -- (202.78,165.88) -- cycle ; 

\path  [shading=_jl11j7z0x,_e2xg8l0yr] (292.78,128.78) -- (329.88,128.78) -- (329.88,165.88) -- (292.78,165.88) -- cycle ; 
 \draw  [color={rgb, 255:red, 208; green, 2; blue, 27 }  ,draw opacity=1 ] (292.78,128.78) -- (329.88,128.78) -- (329.88,165.88) -- (292.78,165.88) -- cycle ; 

\path  [shading=_meb3eig1o,_eq40yeb6f] (538.78,128.78) -- (575.88,128.78) -- (575.88,165.88) -- (538.78,165.88) -- cycle ; 
 \draw  [color={rgb, 255:red, 126; green, 211; blue, 33 }  ,draw opacity=1 ] (538.78,128.78) -- (575.88,128.78) -- (575.88,165.88) -- (538.78,165.88) -- cycle ; 

\path  [shading=_tsi6g1c3q,_idtn18m1f] (380.78,129.78) -- (417.88,129.78) -- (417.88,166.88) -- (380.78,166.88) -- cycle ; 
 \draw  [color={rgb, 255:red, 126; green, 211; blue, 33 }  ,draw opacity=1 ] (380.78,129.78) -- (417.88,129.78) -- (417.88,166.88) -- (380.78,166.88) -- cycle ; 

\path  [shading=_78v7e9k6w,_acwmx1e7t] (453.78,129.78) -- (490.88,129.78) -- (490.88,166.88) -- (453.78,166.88) -- cycle ; 
 \draw  [color={rgb, 255:red, 126; green, 211; blue, 33 }  ,draw opacity=1 ] (453.78,129.78) -- (490.88,129.78) -- (490.88,166.88) -- (453.78,166.88) -- cycle ; 

\draw (313,433) node [anchor=north west][inner sep=0.75pt]   [align=center] {User};
\draw (162,49) node [anchor=center][inner sep=0.75pt]   [align=center] {$\mathbf{A} \ \triangleq \ ( A_{1} ,\ \dotsc ,\ A_{m})$};
\draw (545,49) node [anchor=center][inner sep=0.75pt]   [align=center] {$\mathbf{B} \ \triangleq \ ( B_{1} ,\ \dotsc ,\ B_{m})$};
\draw (270,144) node [anchor=north west][inner sep=0.75pt]   [align=center] {$\displaystyle f_{N}$};
\draw (300,140) node [anchor=north west][inner sep=0.75pt]   [align=center] {$\displaystyle \tilde{\mathbf{A}}_{N}$};
\draw (560,147.33) node   [align=center] {\begin{minipage}[lt]{13.52pt}\setlength\topsep{0pt}
$\displaystyle \mathbf{B}$
\end{minipage}};
\draw (120,261) node   [align=center] {\begin{minipage}[lt]{46.84pt}\setlength\topsep{0pt}
Server $\displaystyle 1$
\end{minipage}};
\draw (334,264) node   [align=center] {\begin{minipage}[lt]{46.84pt}\setlength\topsep{0pt}
Server $\displaystyle i$
\end{minipage}};
\draw (570,267) node   [align=center] {\begin{minipage}[lt]{46.84pt}\setlength\topsep{0pt}
Server $\displaystyle N$
\end{minipage}};
\draw (432,256) node [anchor=north west][inner sep=0.75pt]   [align=center] {$\displaystyle \dotsc $};
\draw (208,255) node [anchor=north west][inner sep=0.75pt]   [align=center] {$\displaystyle \dotsc $};
\draw (174,333) node [anchor=north west][inner sep=0.75pt]   [align=center] {$\displaystyle h_{1}$};
\draw (431,332) node [anchor=north west][inner sep=0.75pt]   [align=center] {$\displaystyle h_{N}$};
\draw (297,333) node [anchor=north west][inner sep=0.75pt]   [align=center] {$\displaystyle h_{i}$};
\draw (201,338) node [anchor=north west][inner sep=0.75pt]   [align=center] {$\displaystyle Z_{1}$};
\draw (323,337) node [anchor=north west][inner sep=0.75pt]   [align=center] {$\displaystyle Z_{i}$};
\draw (461,335) node [anchor=north west][inner sep=0.75pt]   [align=center] {$\displaystyle Z_{N}$};
\draw (229,483) node [anchor=north west][inner sep=0.75pt]   [align=center] {$\displaystyle \mathbf{AB} \ \triangleq \ ( A_{1} B_{1} ,\ \dotsc ,\ A_{m} B_{m})$};
\draw (253,332) node [anchor=north west][inner sep=0.75pt]   [align=center] {$\displaystyle \dotsc $};
\draw (372,334) node [anchor=north west][inner sep=0.75pt]   [align=center] {$\displaystyle \dotsc $};
\draw (179,142) node [anchor=north west][inner sep=0.75pt]   [align=center] {$\displaystyle f_{i}$};
\draw (89,143) node [anchor=north west][inner sep=0.75pt]   [align=center] {$\displaystyle f_{1}$};
\draw (120,140) node [anchor=north west][inner sep=0.75pt]   [align=center] {$\displaystyle \tilde{\mathbf{A}}_{1}$};
\draw (212,140) node [anchor=north west][inner sep=0.75pt]   [align=center] {$\displaystyle \tilde{\mathbf{A}}_{i}$};
\draw (402,148.33) node [align=center] {\begin{minipage}[lt]{13.52pt}\setlength\topsep{0pt}
$\displaystyle \mathbf{B}$
\end{minipage}};
\draw (476,148.33) node   [align=center] {\begin{minipage}[lt]{13.52pt}\setlength\topsep{0pt}
$\displaystyle \mathbf{B}$
\end{minipage}};
\end{tikzpicture}}
\caption{The matrices in \(\mathbf{A}\) are first encoded using the functions \(f_i\). These encoded matrices, denoted as \(\tilde{\mathbf{A}}_i\), are then distributed to each Server \(i \in [N]\). Server~\(i\) processes \(\tilde{\mathbf{A}}_i\) and \(\mathbf{B}\) to compute a response \(Z_i\) using the function \(h_i\). Once the computations are complete, the user collects the responses from the \(k\) fastest servers and reconstructs the product \(\mathbf{AB}\).} \label{fig}
\end{figure}
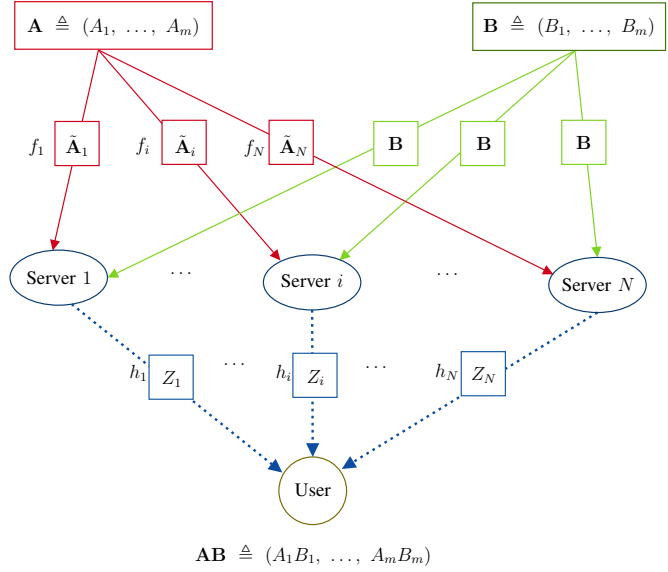

\begin{defn} \label{Def_2}
      For any  $\ell \in [0: k)$ and $\alpha \in [0, 1) \cap \mathbb{Q}$, an $(N, k, r)$-coding scheme is $(\ell, \alpha)$-private\ if
    \begin{align}
        \max_{\mathcal{I} \in [N]^{\geq k}} H(\mathbf{AB}| Z_{\mathcal{I}}) &= 0, 
        \ (Recoverability) \label{recoverability_constraint}
        \\
        \max_{\mathcal{L} \in [N]^{\leq \ell}} \frac{I(\mathbf{A};\tilde{\mathbf{A}}_{\mathcal{L}})}{H(\mathbf{A})} &\leq \alpha.\  (Privacy) \label{privacy_constraint}
    \end{align}
    An achievable rate $\Lambda_{k}$ for an $(N, k, r)$-coding scheme that satisfies \eqref{recoverability_constraint} and \eqref{privacy_constraint} is determined by the ratio 
        \begin{align}	
        	&\Lambda_{k} \triangleq \frac{H(\mathbf{AB}|\mathbf{B})}{\max_{\mathcal{I} \in [N]^{= k}} \sum_{i \in \mathcal{I}} H(Z_{i})}. \label{rate} 
        \end{align}
    The capacity $C_{(\ell, \alpha, k)}$ is the supremum of all achievable rates. Additionally, we define the optimal rate of local randomness~as
    \begin{align*}
        R_{(\ell, \alpha, k)} \triangleq \frac{\min \{ r \in \mathbb{N}: \exists (\ell, \alpha)\text{-private }(N, k, r)\text{-coding scheme}\}}{H(\mathbf{AB}|\mathbf{B})}.
    \end{align*}
\end{defn}
\eqref{recoverability_constraint} means that the responses from any subset of $k$ or more servers are sufficient to reconstruct the product $\mathbf{AB}$.  Note that since conditioning reduces entropy, it is sufficient to take the maximization over $\mathcal{I} \in [N]^{=k}$. 

\eqref{privacy_constraint} means that $\ell$ colluding servers cannot learn more than a fraction $\alpha$ of $\mathbf{A}$ from the encoded matrices $\tilde{\mathbf{A}}_{\mathcal{L}}$. 
Note that $\alpha$ is chosen as a rational number, which is not a restrictive assumption because by density of $\mathbb{Q}$ in $\mathbb{R}$, for any $\beta \in [0, 1]$, $\epsilon > 0$, there exists $\alpha \in [0, 1) \cap \mathbb{Q}$ such that $|\alpha - \beta| < \epsilon$.

In \eqref{rate}, we consider the maximum over any subset of $k$ servers in the denominator to account for the worst-case scenario. Fig.~\ref{fig} illustrates our setting.

\section{Main Results}\label{Main Results}

The following theorem establishes upper and lower bounds on the capacity.

\begin{thm}[Communication Rate] \label{th1}
For any $\alpha \in [0, 1) \cap \mathbb{Q}$ and $\ell \in [0 : k)$, the capacity $C_{(\ell, \alpha, k)}$ satisfies
    \begin{align}
        C_{(\ell, \alpha, k)} &\leq  \left\{\begin{matrix}
            \min \left(\frac{k - \ell}{k\left(1 - \alpha \max\left(1, \frac{D}{E}\right)\right)}, 1\right)& \textnormal{if }\alpha < \frac{1}{\max\left(1, \frac{D}{E}\right)}\\
               1 & \textnormal{if }\alpha \geq \frac{1}{\max\left(1, \frac{D}{E}\right)} 
        \end{matrix}\right., \label{th2_rate_upper_bound} \displaybreak[0]\\
        C_{(\ell, \alpha, k)} &\geq \left\{\begin{matrix} \frac{k - \ell}{k(1 - \alpha)} \min \left(1, \frac{D}{E}\right)\ & \textnormal{if }\alpha < \frac{\ell}{k}\\
             \min\left(1, \frac{D}{E}\right) &\textnormal{if } \alpha \geq \frac{\ell}{k} \end{matrix}\right. \ . \label{th2_rate_lower_bound}
    \end{align}
\end{thm}
\begin{proof}
We prove the converse and achievability in Sections~\ref{Converse} and~\ref{Achievability}, respectively. 
\end{proof}

\begin{figure*}[!t]
\centering
   \begin{minipage}{0.5\textwidth}
   \includegraphics[width=\textwidth]{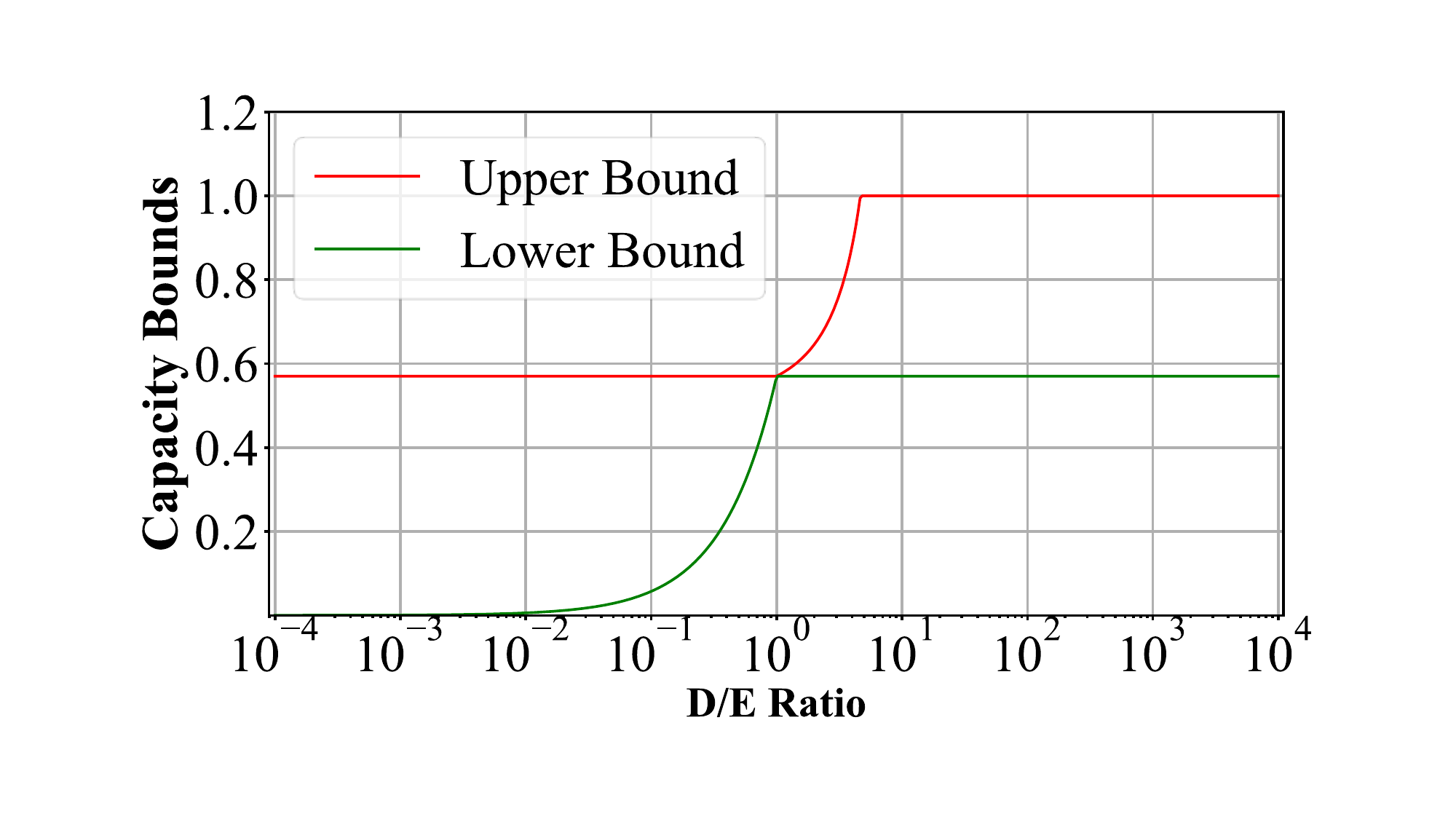}
   \end{minipage}
   \hspace{-9mm}
   \begin{minipage}{0.5\textwidth}
   \includegraphics[width=\textwidth]{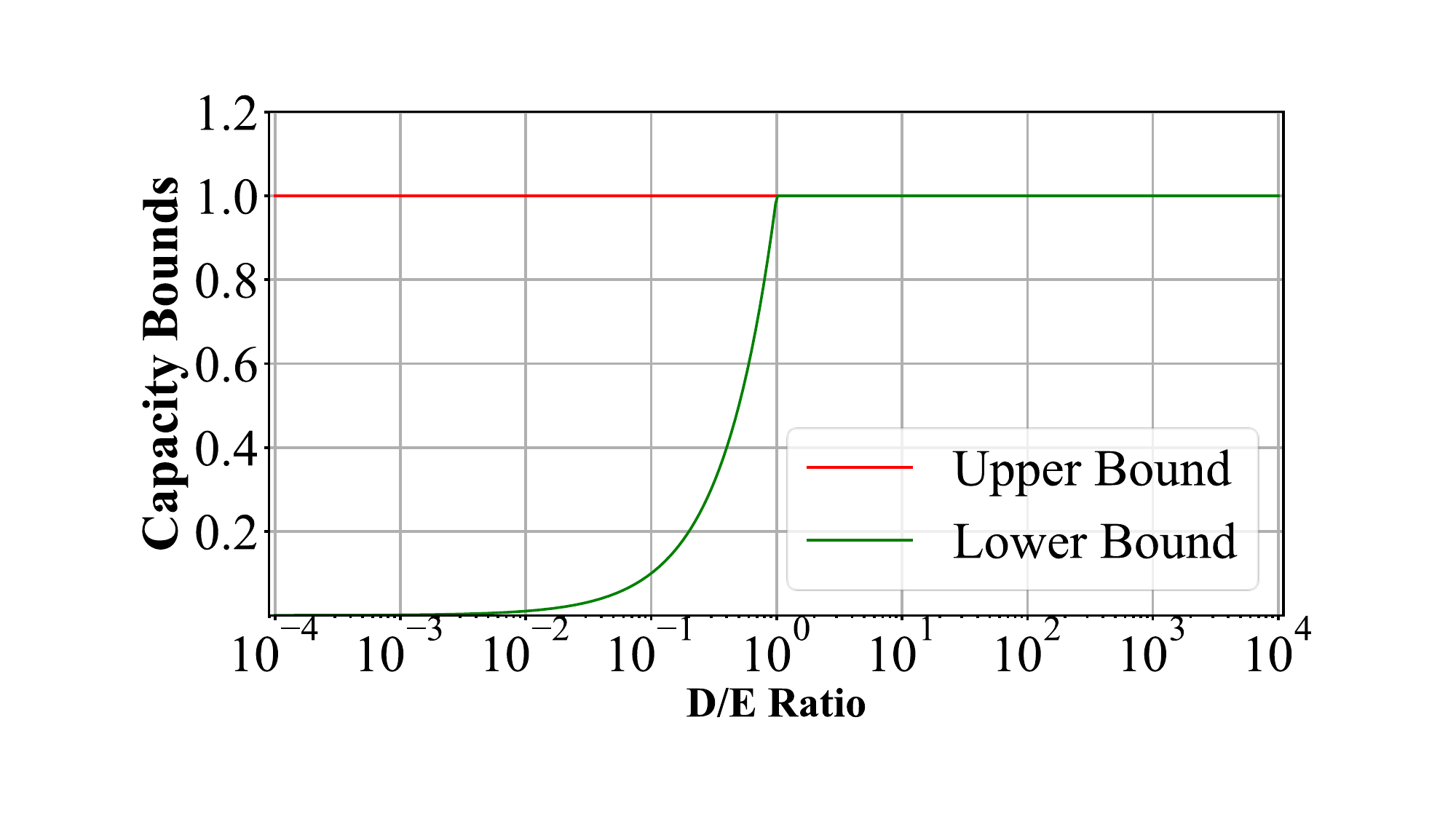}
   \end{minipage}
   \vspace{-3mm}
   \caption{Capacity upper and lower bounds vs. D/E ratio. Left: Case \(\alpha < \ell/k\). Right:  Case \(\alpha \geq \ell/k\).}
   \label{fig:cap_bounds} 
\end{figure*}

\begin{cor} \label{cor_matching_bounds}
When \(D \geq E\) and \(\alpha \geq \ell/k\), then \(C_{(\ell, \alpha, k)} = 1\). This follows from the upper and lower bounds matching in Theorem \ref{th1}. 
\end{cor}

\begin{thm}[Local Randomness] \label{th2}
For any $\alpha \in [0, 1) \cap \mathbb{Q}$ and $\ell \in [1: k)$, the optimal rate of local randomness necessary at the encoder satisfies
    \begin{align}
        R_{(\ell, \alpha, k)} &\leq \left\{\begin{matrix} \frac{\ell - \alpha k}{k - \ell}  \max\left(1, \frac{D}{E}\right) &\textnormal{if } \ \alpha < \frac{\ell}{k} \\
        0 &\textnormal{if } \ \alpha \geq \frac{\ell}{k} \end{matrix}\right. \ , \label{th2_optimal_randomness_upper_bound} \\
        R_{(\ell, \alpha, k)} &\geq \ \frac{{\left[\ell - \alpha k \max\left(1, \frac{D}{E}\right)\right]}^{+}}{k - \ell}.\label{th2_optimal_randomness_lower_bound}
    \end{align}
\end{thm}
\begin{proof}
We prove the converse and achievability in Sections~\ref{local_randomness_converse} and~\ref{Randomness}, respectively.
\end{proof}
\begin{cor} 
     From Theorem~\ref{th2}, when $\alpha \geq \frac{\ell}{k}$, then $R_{(\ell, \alpha, k)} = 0$. When $\alpha < \frac{\ell}{k}$ and $\frac{D}{E} \leq 1$, then $R_{(\ell, \alpha, k)} =  \frac{\ell - \alpha k}{k - \ell}  $.
\end{cor}
The bounds established in the previous theorems match when the matrices are square.
\begin{thm}[Optimality Results] \label{th3}
    If $\mathbf{A}$ and $\mathbf{B}$ are two sequences of independent and square matrices with uniformly distributed entries, then the capacity is
    \begin{align}
        C_{(\ell, \alpha, k)} = \min\left(\frac{k - \ell}{k(1 - \alpha)}, 1\right) = \left\{\begin{matrix} \frac{k - \ell}{k(1 - \alpha)} &\textnormal{if } \alpha < \frac{\ell}{k} \\
            1 &\textnormal{if } \alpha \geq \frac{\ell}{k} \end{matrix}\right. \ , \nonumber
    \end{align}
     and the optimal rate of local randomness is
    \begin{align}
        R_{(\ell, \alpha, k)} = \frac{{[\ell - \alpha k]}^{+}}{k - \ell} = \left\{\begin{matrix} \frac{\ell - \alpha k}{k - \ell} &\textnormal{if } \alpha < \frac{\ell}{k} \\
            0 &\textnormal{if } \alpha \geq \frac{\ell}{k} \end{matrix}\right. \ . \nonumber
    \end{align}
\end{thm}
\begin{proof}
One can deduce these results from the bounds in Theorems~\ref{th1} and~\ref{th2} with $D = E$.    
\end{proof}
Fig.~\ref{fig:cap_bounds}  illustrates Theorem \ref{th1} and shows that the capacity is characterized when $D/E=1$, as stated in Theorem \ref{th3}, or when \(D \geq E\) and \(\alpha \geq \ell/k\), as stated in Corollary \ref{cor_matching_bounds}. Let $gap$ be the difference between the upper bound in \eqref{th2_rate_upper_bound} and the lower bound in \eqref{th2_rate_lower_bound}. When $\alpha < \frac{\ell}{k}$, if $\frac{D}{E} \rightarrow 0$, then $gap \rightarrow \frac{k-\ell}{k(1 - \alpha)}$, and if $\frac{D}{E} \rightarrow \infty$, then $gap \rightarrow 1 - \frac{k-\ell}{k(1 - \alpha)}$. When $\alpha \geq \frac{\ell}{k}$, if $\frac{D}{E} \rightarrow 0$, then $gap \rightarrow 1$, and if $\frac{D}{E} \geq 1$, then $gap = 0$. 
\section{Converse Proof} \label{Converse}
Define $\mathcal{L} \subseteq [N]$ such that  $|\mathcal{L}| = \ell $.
\subsection{Rate}
Consider $\alpha < \frac{1}{\max(1, \frac{D}{E})}$. We will use the following lemmas. 
\begin{lem} \label{lm0}
    For any $\mathcal{I} \subseteq [N]$, we have
    \begin{align}
    I(\mathbf{A}; \tilde{\mathbf{A}}_{\mathcal{I}}, \mathbf{B}) &\geq I( Z_{\mathcal{I}};\mathbf{A}|\mathbf{B}), \label{converse_proof_lm0-1} \\
    I(\mathbf{A}; Z_{\mathcal{I}}|\mathbf{B}) &\geq I(\mathbf{AB}; Z_{\mathcal{I}}|\mathbf{B}). \label{converse_proof_lm0-2} 
    \end{align} 
\end{lem}
\begin{proof}
    For \eqref{converse_proof_lm0-1}, we have
    \begin{align}
        I(\tilde{\mathbf{A}}_{\mathcal{I}}, \mathbf{B}; \mathbf{A}) &\geq I(Z_{\mathcal{I}}, \mathbf{B};\mathbf{A}) \geq I(Z_{\mathcal{I}};\mathbf{A}|\mathbf{B}), \label{converse_proof_b_1}
    \end{align}
     where the first inequality holds because $Z_{\mathcal{I}}$ is a function of $\mathbf{B}$ and $\tilde{\mathbf{A}}_{\mathcal{I}}$, and the second inequality holds by the chain rule and non-negativity of mutual information. For \eqref{converse_proof_lm0-2}, we have
    \begin{align}
        I(Z_{\mathcal{I}}; \mathbf{AB}|\mathbf{B}) &= H(Z_{\mathcal{I}}|\mathbf{B}) - H(Z_{\mathcal{I}}|\mathbf{AB},\mathbf{B}) \nonumber 
        \\ &\leq H(Z_{\mathcal{I}}|\mathbf{B}) - H(Z_{\mathcal{I}}|\mathbf{AB},\mathbf{B}, \mathbf{A}) \nonumber 
        \\ &= H(Z_{\mathcal{I}}|\mathbf{B}) - H(Z_{\mathcal{I}}|\mathbf{B}, \mathbf{A}) \nonumber 
        \\ &= I(Z_{\mathcal{I}};\mathbf{A}|\mathbf{B}), \nonumber
    \end{align}
where the inequality holds because conditioning reduces entropy.
\end{proof}
The following lemma bounds the mutual information between the observations of a set $\mathcal{L}$ of colluding servers and the product $\mathbf{AB}$.  
\begin{lem} \label{lm1}
 For any $\alpha \in [0, 1) \cap \mathbb{Q}$, we have
    \begin{align}
         I(Z_{\mathcal{L}}; \mathbf{AB}|\mathbf{B}) \leq \alpha H(\mathbf{A}). \nonumber
    \end{align} 
\end{lem}

\begin{proof}
    We have
    \begin{align} 
    	I(Z_{\mathcal{L}}; \mathbf{AB}|\mathbf{B})
            & \overset{\mathrm{(a)}}{\leq} I(Z_{\mathcal{L}}; \mathbf{A}|\mathbf{B}) \nonumber\\
    	& \overset{\mathrm{(b)}}{\leq} I(\tilde{\mathbf{A}}_{\mathcal{L}},\mathbf{B}; \mathbf{A}) \nonumber \\
    	& = I(\tilde{\mathbf{A}}_{\mathcal{L}}; \mathbf{A}) + I(\mathbf{B}; \mathbf{A}|\tilde{\mathbf{A}}_{\mathcal{L}}) \nonumber \\
    	& \overset{\mathrm{(c)}}{=} I(\tilde{\mathbf{A}}_{\mathcal{L}}; \mathbf{A}) \nonumber\displaybreak[0] \\
    	& \overset{\mathrm{(d)}}{\leq} \alpha H(\mathbf{A}), \nonumber
    \end{align} 
where
    \begin{enumerate}
        \item [(a)] holds by \eqref{converse_proof_lm0-2};
        \item [(b)] holds by \eqref{converse_proof_lm0-1};
        \item [(c)] holds because 
        \begin{align}
            0 &= I(\mathbf{B}; \mathbf{A}, \tilde{\mathbf{A}}_{\mathcal{L}}) \nonumber \\
            &= I(\mathbf{B}; \tilde{\mathbf{A}}_{\mathcal{L}}) + I(\mathbf{B};\mathbf{A}|\tilde{\mathbf{A}}_{\mathcal{L}}) \nonumber \\
            &\geq I(\mathbf{B};\mathbf{A}|\tilde{\mathbf{A}}_{\mathcal{L}}), \nonumber
        \end{align}
        where the first equality holds by the independence between $\mathbf{A}$ and $\mathbf{B}$ and the definition of $\tilde{\mathbf{A}}_{\mathcal{L}}$;
        \item [(d)] holds by \eqref{privacy_constraint}.
    \end{enumerate}
\end{proof}   

Then, for all $\mathcal{I} \in [N]^{= k}$ such that $\mathcal{L} \subseteq \mathcal{I}$, we have
\begin{align} 
    H(\mathbf{AB}|\mathbf{B})
    &= H(\mathbf{AB}|\mathbf{B}) -  H(\mathbf{AB}|Z_{\mathcal{I}},\mathbf{B}) + H(\mathbf{AB}|Z_{\mathcal{I}},\mathbf{B}) \nonumber \\
    &\overset{\mathrm{(a)}}{=} I(Z_{\mathcal{I}}; \mathbf{AB}|\mathbf{B}) \nonumber \\
    &= H(Z_{\mathcal{I}}|\mathbf{B}) - H(Z_{\mathcal{I}}|\mathbf{AB},\mathbf{B}) \nonumber \\
    &\overset{\mathrm{(b)}}{\leq} H(Z_{\mathcal{I}}|\mathbf{B}) - H(Z_{\mathcal{L}}|\mathbf{AB},\mathbf{B}) \label{Redo_Converse}\\
    &\overset{\mathrm{(c)}}{\leq} H(Z_{\mathcal{I}}|\mathbf{B}) - H(Z_{\mathcal{L}}|\mathbf{B}) + {\alpha}H(\mathbf{A}),\label{converse_1}
\end{align}
where
\begin{enumerate}
    \item [(a)] holds because $ H(\mathbf{AB}|Z_{\mathcal{I}},\mathbf{B}) \leq  H(\mathbf{AB}|Z_{\mathcal{I}}) = 0$ by \eqref{recoverability_constraint};
    \item [(b)] holds because $\mathcal{L} \subseteq \mathcal{I}$;
    \item [(c)] holds by Lemma~\ref{lm1}.     
\end{enumerate}
Then, we have 
\begin{align} 
    H(\mathbf{AB}|\mathbf{B})
    &\overset{\mathrm{(a)}}{\leq} H(Z_{\mathcal{I}}|\mathbf{B}) - \ell \frac{1}{\binom{k}{\ell}}\sum_{\mathcal{L} \in \mathcal{I} ^{= \ell}} \frac{H(Z_{\mathcal{L}}|\mathbf{B})}{\ell} + {\alpha}H(\mathbf{A}) \nonumber \displaybreak[0] \\
    &\overset{\mathrm{(b)}}{\leq} H(Z_{\mathcal{I}}|\mathbf{B}) - \ell \frac{H(Z_{\mathcal{I}}|\mathbf{B})}{k} + {\alpha}H(\mathbf{A}) \nonumber \displaybreak[0] \\
    &= \left(1 - \frac{\ell}{k}\right) H(Z_{\mathcal{I}}|\mathbf{B}) + {\alpha}H(\mathbf{A}) \nonumber\\
    &\leq \left(1 - \frac{\ell}{k}\right) H(Z_{\mathcal{I}}) + {\alpha}H(\mathbf{A}) \nonumber \\
    &\leq \left(1 - \frac{\ell}{k}\right) \sum_{i \in \mathcal{I}}H(Z_i) + {\alpha}H(\mathbf{A}), \label{converse_proof}
\end{align}
where
\begin{enumerate}
    \item[(a)] holds by averaging \eqref{converse_1} over all possible subsets $\mathcal{L}$ of servers of size $\ell$ in $\mathcal{I}$;
    \item[(b)] holds by Han's inequality \cite[Section 17.6]{cover1999elements}.
\end{enumerate}

Then, from \eqref{converse_proof}, we have
\begin{align}
    \Lambda_{k}
    &= \frac{H(\mathbf{AB}|\mathbf{B})}{{\max_{\mathcal{I} \in [N]^{= k}}}\sum_{i \in \mathcal{I}} H(Z_i)} \nonumber \\
    &\leq \frac{1 - \frac{\ell}{k}}{1 - {\alpha}\frac{H(\mathbf{A})}{H(\mathbf{AB}|\mathbf{B})}}. \label{converse_rate_1}
\end{align}
Since the matrices in $\mathbf{A}$ and $\mathbf{B}$ are independent and uniformly distributed over $\mathbb{F}_{q}^{C \times D}$ and $\mathbb{F}_q^{D \times E}$, by \cite[Lemma~2]{jia2021capacity}, we have
    \begin{align}
    q \rightarrow \infty \Rightarrow H(\mathbf{AB}|\mathbf{B}) =  m \times \min (CD, CE).\label{use_lemma_2_b}
    \end{align}
Also, for any $\alpha \in [0,1)$,  we have
    \begin{align} 
        H(\mathbf{AB}|\mathbf{B})\ &\overset{\mathrm{(a)}}{\leq}\ H(Z_{\mathcal{I}}|\mathbf{B}) - H(Z_{\mathcal{L}}|\mathbf{AB},\mathbf{B}) \nonumber \displaybreak[0]\\
        &\leq H(Z_{\mathcal{I}}|\mathbf{B}) \nonumber \\
        &\overset{\mathrm{(b)}}{\leq} \sum_{i \in \mathcal{I}}H(Z_i), \label{converse_proof_0}
    \end{align}
    where 
    \begin{enumerate}
        \item [(a)] holds by \eqref{Redo_Converse};
        \item [(b)] holds because conditioning reduces entropy.
    \end{enumerate}
    Hence, from \eqref{converse_proof_0}, we have 
    \begin{align}
        \Lambda_{k}
        &=\frac{H(\mathbf{AB}|\mathbf{B})}{{\max_{\mathcal{I} \in [N]^{=k}}}\sum_{i \in \mathcal{I}} H(Z_i)} \nonumber \\
        &{\leq} \frac{ \sum_{i \in \mathcal{I}}H(Z_i)}{{\max_{\mathcal{I} \in [N]^{=k}}}\sum_{i \in \mathcal{I}} H(Z_i)} \nonumber \\
        &{\leq}\ 1. \label{rate_upperbound_3}
    \end{align}
Finally, we obtain \eqref{th2_rate_upper_bound} from \eqref{converse_rate_1}, \eqref{use_lemma_2_b} and \eqref{rate_upperbound_3}.

\subsection{Local Randomness} \label{local_randomness_converse}
Consider $\alpha \in [0, 1) \cap \mathbb{Q}$, $\ell \in [1: k)$, and an $(\ell, \alpha)$-private $(N, k, r)$-coding scheme. 

\begin{lem} \label{LR-lm2}
Let $\mathcal{V} \subseteq \mathcal{L}$ and $v \triangleq |\mathcal{V}|$. For $\mathcal{W} \subseteq [N] \setminus \mathcal{V}$ and $S \subseteq [N] \setminus (\mathcal{W} \cup \mathcal{V})$ such that $|\mathcal{W}| = \ell - v$ and $|\mathcal{S}| = k - \ell$, we have
    \begin{align}
        &\frac{1 -\alpha \frac{H(\mathbf{A})}{H(\mathbf{AB}|\mathbf{B})}}{k - \ell}H(\mathbf{AB}|\mathbf{B}) \nonumber \\
        &\leq H(\tilde{\mathbf{A}}_{[N]}|\tilde{\mathbf{A}}_{\mathcal{V}}) - H(\tilde{\mathbf{A}}_{[N]}|\tilde{\mathbf{A}}_{\mathcal{V} \cup \{i ^{*} (\mathcal{V})\}}), \label{randomness_converse_2}
    \end{align}
where $i^{*}(\mathcal{V}) \in \arg\max_{i \in [N] \setminus \mathcal{V}} H(\tilde{\mathbf{A}}_{i}|\tilde{\mathbf{A}}_{\mathcal{V}})$.
\end{lem}
\begin{proof}
     See Appendix~\ref{appendix-1}.   
\end{proof}
Next, we define $\mathcal{V}_{0} \triangleq \emptyset$, and for $j \in \mathcal{L}$, $\mathcal{V}_{j} \triangleq \mathcal{V}_{j-1} \cup {\{i^{*}(\mathcal{V}_{j-1})}\}$. Then, we have 
\begin{align}
    &\frac{\ell - k\alpha \frac{H(\mathbf{A})}{H(\mathbf{AB}|\mathbf{B})}}{k - \ell}H(\mathbf{AB}|\mathbf{B})  \nonumber \\
    &= \left(\ell\frac{1 - \alpha \frac{H(\mathbf{A})}{H(\mathbf{AB}|\mathbf{B})}}{k - \ell} - \alpha \frac{H(\mathbf{A})}{H(\mathbf{AB}|\mathbf{B})} \right)H(\mathbf{AB}|\mathbf{B}) \nonumber \displaybreak[0]\\
    &= -\alpha H(\mathbf{A}) + \ell \frac{1 - \alpha \frac{H(\mathbf{A})}{H(\mathbf{AB}|\mathbf{B})}}{k - \ell}H(\mathbf{AB}|\mathbf{B}) \nonumber \displaybreak[0]\\
    &= -\alpha H(\mathbf{A}) + \sum_{i = 0}^{\ell - 1}\frac{1 - \alpha \frac{H(\mathbf{A})}{H(\mathbf{AB}|\mathbf{B})}}{k - \ell}H(\mathbf{AB}|\mathbf{B}) \nonumber \displaybreak[0]\\
    &\overset{(a)}{\leq} -\alpha H(\mathbf{A}) + \sum_{i = 0}^{\ell - 1}\left[ H(\tilde{\mathbf{A}}_{[N]}|\tilde{\mathbf{A}}_{\mathcal{V}_{i}}) - H(\tilde{\mathbf{A}}_{[N]}|\tilde{\mathbf{A}}_{\mathcal{V}_{i+1}})\right] \nonumber \displaybreak[0]\\
    &= -\alpha H(\mathbf{A}) + H(\tilde{\mathbf{A}}_{[N]}) - H(\tilde{\mathbf{A}}_{[N]}|\tilde{\mathbf{A}}_{\mathcal{V}_{\ell}}) \nonumber \displaybreak[0]\\
    &\overset{(b)}{\leq} -\alpha H(\mathbf{A}) + H(\mathbf{A}, R) - H(\tilde{\mathbf{A}}_{[N]}|\tilde{\mathbf{A}}_{\mathcal{V}_{\ell}}) \nonumber \displaybreak[0]\\
    &\overset{(c)}{=} (1 - \alpha)H(\mathbf{A}) + H(R) - H(\tilde{\mathbf{A}}_{[N]}|\tilde{\mathbf{A}}_{\mathcal{V}_{\ell}}) \nonumber \displaybreak[0]\\
    &\overset{(d)}{=} (1 - \alpha)H(\mathbf{A}) + H(R) - H(\mathbf{A},\tilde{\mathbf{A}}_{[N]}|\tilde{\mathbf{A}}_{\mathcal{V}_{\ell}}) \nonumber \displaybreak[0]\\
    &\leq (1 - \alpha)H(\mathbf{A}) + H(R) - H(\mathbf{A}|\tilde{\mathbf{A}}_{\mathcal{V}_{\ell}}) \nonumber \displaybreak[0]\\
    &\overset{(e)}{\leq} H(R), \label{randomness_converse_3}
\end{align}
where
\begin{enumerate}
    \item [(a)] holds by applying $\ell$ times \eqref{randomness_converse_2} and the definition of $\mathcal{V}_{j}$, $j \in \mathcal{L}$;
    \item [(b)] holds because $\tilde{\mathbf{A}}_{[N]}$ is a deterministic function of $(\mathbf{A}, R)$;
    \item [(c)] holds by independence between $\mathbf{A}$ and $R$;
    \item [(d)] holds because by \eqref{recover_all} we have 
\begin{align}   H(\mathbf{A},\tilde{\mathbf{A}}_{[N]}|\tilde{\mathbf{A}}_{\mathcal{V}_{\ell}})
&= H(\mathbf{A}|\tilde{\mathbf{A}}_{[N]}) + H(\tilde{\mathbf{A}}_{[N]}|\tilde{\mathbf{A}}_{\mathcal{V}_{\ell}}) \nonumber \\
&= H(\tilde{\mathbf{A}}_{[N]}|\tilde{\mathbf{A}}_{\mathcal{V}_{\ell}}); \nonumber
\end{align}
    \item [(e)] holds because $-H(\mathbf{A}|\tilde{\mathbf{A}}_{\mathcal{V_{\ell}}}) \leq -(1 - \alpha)H(\mathbf{A})$ by \eqref{privacy_constraint}.
\end{enumerate}
Using \eqref{randomness_converse_3} and \eqref{use_lemma_2_b}, we have
    \begin{align}
         \frac{\ell -k\alpha \max(1, \frac{D}{E})}{k - \ell}H(\mathbf{AB}|\mathbf{B}) \leq H(R). \label{randomness_converse_4}
    \end{align}
Finally, \eqref{th2_optimal_randomness_lower_bound} holds by \eqref{randomness_converse_4}, and since $H(R) \geq 0$ is always true.

\section{Achievability Proof}\label{Achievability}

For the achievability, the idea is to design a coding scheme with the following leakage symmetry, 
\begin{align}
&\forall t \in [N], \exists E_t \in [0, 1], \forall\ \mathcal{I} \subseteq [N],\nonumber\\&|\mathcal{I}| = t \Rightarrow \frac{I(\mathbf{A};\tilde{\mathbf{A}}_{\mathcal{I}})}{H(\mathbf{A})} = E_t, \label{access_function_def}
\end{align}
which means that the leakage of any set of encoded matrices $\tilde{\mathbf{A}}_{\mathcal{I}} \triangleq (\tilde{\mathbf{A}}_{i})_{i\in \mathcal{I}}$ only depends on the cardinality of $\mathcal{I}$. Consequently, the amount of information leakage of $\mathbf{A}$ can be fully described by the following function.
\begin{align} 
    g :[N] \to [0, 1],\ t \mapsto E_t. \label{leakage_symmetry}
\end{align} 
The recoverability and privacy constraints  \eqref{recoverability_constraint} and \eqref{privacy_constraint} impose the following constraints on $g$: 
\begin{align}
    &\forall t \in [\ell],\ g(t) \leq \alpha, \label{g_less_than_alpha} \\
    & \forall t \in [k: N],\  g(t) = 1. \nonumber
\end{align}
Consider two cases: $\alpha \geq \frac{\ell}{k}$ and $\alpha < \frac{\ell}{k}$. The first case will be handled with a modified ramp secret-sharing scheme that will yield $g$ defined as
\begin{align}
	& g: i \mapsto \left\{\begin{matrix}
		& \frac{i}{k}, & \ i\ \in [0: k)\\
		& 1 ,& i\ \in [k: N] 
	\end{matrix} \right.\ .\label{case_1_g} 
\end{align}
For the case $\alpha < \frac{\ell}{k}$, we define $g$ as \begin{align}
    	&g = g_1 + g_2\label{g_def},
\end{align} with
    \begin{align}
    	&g_1: i \mapsto \left\{\begin{matrix}
    		& \frac{\alpha}{\ell}i,\ &i\ \in [0: k)\\
    		& \frac{\alpha}{\ell}k,\ &i\ \in [k: N]
    	\end{matrix}\right.\ ,\label{g1_def}
        \\
    	&g_2: i \mapsto \left\{\begin{matrix}
    		& 0,\ & i \in [0: \ell]\\
    		& \frac{1- \alpha}{k - \ell}(i - \ell) + \alpha - \frac{\alpha}{\ell}i,\ & i \in (\ell:\ k)\\
    		& 1 - \frac{\alpha}{\ell}k,\ & i \in [k: N]
    	\end{matrix}\right. \ , \label{g2_def}
    \end{align}
and construct a coding scheme by combining two modified ramp secret-sharing schemes that will yield $g_1$ and $g_2$.

In Section~\ref{CSC}, we present our coding scheme. Then, in Section~\ref{ACS}, we analyze our coding scheme and show that it satisfies our setting constraints, \eqref{recoverability_constraint} and \eqref{privacy_constraint}. 

\subsection{Coding Scheme}\label{CSC}

\subsubsection{Case 1} $\alpha \geq \frac{\ell}{k}$. \label{CSC_1} 
We divide the sequences of $m$ matrices \(\mathbf{A}\) and \(\mathbf{B}\) into $J  \triangleq \left\lceil\frac{m}{k}\right\rceil$ blocks of size $k$.  For any $b \in \left[J\right]$, define
\begin{align}
    \mathcal{S}_b \triangleq [(b-1)k + 1 : b k] , \label{S_b_def_0}
\end{align}
and define the corresponding block of $k$ matrices from $\mathbf{A}$ as  
\begin{align}
    A^{ \mathcal{S}_b} \triangleq \left[A_{(b-1)k+1},\ \ldots,\ A_{bk}\right] \in {\left(\mathbb{F}_{q}^{C \times D}\right)}^{1 \times k}. \label{A_S_b_0}
\end{align}
Consider distinct non-zero constants $x_i \in  \mathbb{F}_q, i \in [N]$, and define 
\begin{align}
             M_{i}^{\mathcal{S}_b} \triangleq \begin{bmatrix}
                x_{i}^{(b-1)k} \\
             \vdots \\
               x_{i}^{bk - 1}
            \end{bmatrix} \in \mathbb{F}_{q}^{k \times 1},\nonumber
    \end{align} 
and for $\mathcal{I} \in [N]^{=k}$, \begin{align}
     M_{\mathcal{I}}^{\mathcal{S}_b} \triangleq   (M_{i}^{\mathcal{S}_b})_{i \in \mathcal{I}} \in \mathbb{F}_{q}^{k \times |\mathcal{I}|}. \label{M_I_S0}
\end{align}
\textcolor{black}{We divide  $A_s$, $s \in \mathcal{S}_b$, into $k$ submatrices vertically, i.e.,}
\begin{align}
    \textcolor{black}{A_s \triangleq \Trans(\left[(A_{(s, j)})_{j \in [k]} \right]),} \nonumber
\end{align}
\textcolor{black}{where $A_{(s, j)} \in \mathbb{F}_{q}^{\frac{C }{k}\times D}$.} Then, for any $b \in [J]$, define
\begin{align}
    \textcolor{black}{\forall s \in \mathcal{S}_b, \tilde{A}_{(i,s)}\triangleq \sum_{j \in [k]}x_i^{(b-1)k+j - 1}  A_{(s, j)}.}  \label{A_S_bar_0}
\end{align}
For any $s \in \mathcal{S}_b$, $\tilde{A}_{(i,s)}$ is transmitted to Server~$i \in [N]$. 

Then, define 
  \begin{align}
    	\textcolor{black}{\forall b \in \left[  J \right], s \in \mathcal{S}_b, Z_i^{s}} & \textcolor{black}{\triangleq  \tilde{A}_{(i,s)}B_s \in \mathbb{F}_{q}^{\frac{C }{k}\times E}.}\label{Z_S_0}
    \end{align}
The response for each block $\mathcal{S}_b$ from Server~$i \in [N]$ is
\begin{align}
    \forall b \in \left[  J \right], Z_i^{\mathcal{S}_b} \textcolor{black}{\triangleq (Z_i^{s})_{s \in \mathcal{S}_b} \in \left(\mathbb{F}_{q}^{\frac{C}{k} \times E}\right)^{1 \times k},} \label{Z_S_0_1}
\end{align}
and the total response from Server~$i \in [N]$ is
    \begin{align}
          Z_i &\triangleq (Z_i^{\mathcal{S}_b})_{b \in \left[J\right]} \label{Z_i_0} .
    \end{align}

For any $b \in [J]$, the user downloads the answers from $k$ servers, then recovers $(A_s \times B_s)_{s \in \mathcal{S}_b}$ from

\begin{align}
    	{Z}^{\mathcal{S}_b}_{\mathcal{I}} &\triangleq\ ({Z_i}^{\mathcal{S}_b})_{i \in \mathcal{I}} =  \textcolor{black}{\left(\begin{bmatrix}
    			Z_{1}^{s}\\
    			Z_{2}^{s} \\
    			\vdots\\
    			Z_{k}^{s} 
    		\end{bmatrix}\right)_{s \in \mathcal{S}_b}} =\Trans\left(M_{\mathcal{I}}^{\mathcal{S}_b}\right) \textcolor{black}{\left(\begin{bmatrix}
    			A_{(s,1)}B_{s}\\
    			A_{(s,2)}B_{s} \\
    			\vdots\\
    			A_{(s,k)}B_{s} 
    		\end{bmatrix}\right)_{s \in \mathcal{S}_b},} \label{Z_def_0} 
\end{align}
where we consider $\mathcal{I} = [k]$, without loss of generality. The matrix equation in \eqref{Z_def_0} has a unique solution because $ M_{\mathcal{I}}^{\mathcal{S}_b}$  defined in \eqref{M_I_S0} is invertible since its determinant is a minor of a Vandermonde matrix, for which each column $i$ can be factored by $x_{i}^{(b-1)k}$, and the $(x_i)_{i \in [N]}$ are non-zero and distinct.  
 
\subsubsection{Case 2} $\alpha < \frac{\ell}{k}$. \label{CSC_2} 
 Define the following index sets
\begin{align}
	&\mathcal{P} \triangleq [p], p \triangleq \left\lceil \alpha\frac{k}{\ell}m\right\rceil   \label{divide_index},\\
	&\overline{\mathcal{P}} \triangleq [m] \setminus \mathcal{P}, |\overline{\mathcal{P}}| = m - p. \label{divide_index_2}
\end{align}
We partition $\mathbf{A}$ into two sub-sequences $A^{\mathcal{P}} \triangleq (A_j)_{j \in \mathcal{P}}$ and ${A}^{\overline{\mathcal{P}}} \triangleq (A_{j})_{j \in \overline{\mathcal{P}}}$, then proceed as follows.

\begin{enumerate}
    \item[(i)] Break down the sequences  $A^{\mathcal{P}}$ and ${B}^{\mathcal{P}}$ into blocks of $k$ matrices. 
The encoding is similar to \eqref{A_S_bar_0} but with $A^{\mathcal{P}}$ divided into $J  \triangleq \left\lceil\frac{p}{k}\right\rceil$ blocks. 
Similarly, the response for each block $\mathcal{S}_b$ from Server~$i \in [N]$ is  
\begin{align}
    	\forall b \in \left[  J \right],  \textcolor{black}{ Z_i^{\mathcal{S}_b} = (\tilde{A}_{(i,s)}B_s)_{s \in \mathcal{S}_b}.}\label{Z_S_1}
    \end{align}
The response for the sub-sequence ${\mathcal{P}}$ from Server~$i \in [k]$~is
    \begin{align}
          Z_i^{{\mathcal{P}}} &\overset{\mathrm{\Delta}}{=} (Z_i^{\mathcal{S}_b})_{b \in \left[J\right]} \label{Z_i_1} .
    \end{align}
    The decoding is the same as in the previous case. 
    \item [(ii)]  Break down the sequence of matrices in ${A}^{\overline{\mathcal{P}}}$ and ${B}^{\overline{\mathcal{P}}}$ into $J \triangleq \left\lceil\frac{m - p}{L}\right\rceil$ blocks of size $L \triangleq k -\ell$ matrices.  For any $b \in \left[J\right]$, define
\begin{align}
    \mathcal{S}_b &\triangleq [(b-1)L + 1 : bL].\label{S_b_def}
\end{align}
Then, we define the corresponding block of $L$ matrices from the sequence ${A}^{\overline{\mathcal{P}}}$ as
\begin{align}
    A^{ \mathcal{S}_b} \triangleq \left[A_{(b-1)L+1},\ \ldots,\ A_{bL}\right] \in {\left(\mathbb{F}_{q}^{C \times D}\right)}^{1 \times L}.  \label{A_P_bar}
\end{align}

Consider distinct non-zero constants $x_i \in  \mathbb{F}_q, i \in [N]$, and define 

\begin{align}
        M_{i}^{\mathcal{S}_b} \triangleq \begin{bmatrix}
                x_{i}^{(b-1)L} \\
                  x_{i}^{(b-1)L+1} \\
                 \vdots  \\
                x_{i}^{(b-1)L + k -1} \\
            \end{bmatrix} \in \mathbb{F}_{q}^{k \times 1}.\label{M_i}
    \end{align} 
Then, for any $\mathcal{I} \subseteq [N]$, $b \in [J]$, define
\begin{align}
    M_{\mathcal{I}}^{\mathcal{S}_b} &\triangleq (M_{i}^{\mathcal{S}_b})_{i \in \mathcal{I}}\in \mathbb{F}_{q}^{k\times |\mathcal{I}|}. \label{M_I_Sb}
\end{align}

\textcolor{black}{For any $s \in \mathcal{S}_b$, let's divide each $A_s$ into $L$ submatrices vertically, i.e.,}
\textcolor{black}{\begin{align}
    A_s &\triangleq \Trans(\left[(A_{(s, j)})_{j \in [L]}\right]), \nonumber 
\end{align}}
\textcolor{black}{where $A_{(s, j)} \in \mathbb{F}_{q}^{\frac{C}{L}\times D}$.  
For any $b \in [J]$, let $R^{\mathcal{S}_b} \triangleq \left(R_s\right)_{s \in \mathcal{S}_b}$ be a sequence of independent matrices such that}
\begin{align}
    R_{s} &\triangleq \Trans(\left[(R_{(s, r)})_{r \in [\ell]}\right]), \label{random_mat}
\end{align}
\textcolor{black}{where $R_{(s, r)} \in \mathbb{F}_{q}^{\frac{C}{L}\times D}$}.

 Then, for any $b \in [J]$, define 
\begin{align}
   &\forall s \in \mathcal{S}_b, \tilde{A}_{(i,s)} \triangleq \textcolor{black}{\sum_{j \in [L]} x_{i}^{(b-1)L + j - 1}A_{(s,j)} + \sum_{r \in [\ell]}{x_{i}}^{bL + r - 1}R_{(s, r)}}, \label{A_S_bar}
\end{align}
For any $s \in \mathcal{S}_b$, $\tilde{A}_{(i,s)}$ is transmitted to Server~$i \in [N]$. 

The response from Server~$i \in [N]$ is  
  \begin{align}
    	\textcolor{black}{\forall b \in \left[  J \right], s \in \mathcal{S}_b, Z_i^{s}  \triangleq \tilde{A}_{(i,s)}B_s \in \mathbb{F}_{q}^{\frac{C }{L}\times E}.} \nonumber
    \end{align}
We also define
\begin{align}
     \textcolor{black}{ Z_i^{\mathcal{S}_b} \overset{\mathrm{\Delta}}{=} (Z_i^{s})_{s \in \mathcal{S}_b}\in \left(\mathbb{F}_{q}^{\frac{C}{L} \times E}\right)^{1 \times L},} \label{Z_S_2} 
\end{align}
and
\begin{align}
      Z_i^{\overline{\mathcal{P}}} &\overset{\mathrm{\Delta}}{=} (Z_i^{\mathcal{S}_b})_{b \in \left[J\right]} \label{Z_i_2} .
\end{align}

For any block $b \in [J]$, upon receiving answers from $k$ servers, the user recovers $(A_s \times B_s)_{s \in \mathcal{S}_b}$ from  

\begin{align}
    	{Z}^{\mathcal{S}_b}_{\mathcal{I}} &\triangleq\  \textcolor{black}{({Z_i}^{\mathcal{S}_b})_{i \in \mathcal{I}} =  \left(\begin{bmatrix}
                Z_{1}^{s}\\
                \vdots \\
                Z_{L}^{s} \\
                Z_{L+1}^{s} \\
                \vdots \\
                Z_{k}^{s}
            \end{bmatrix}\right)_{s \in \mathcal{S}_b}}
             =  \Trans\left(M_{\mathcal{I}}^{\mathcal{S}_b}\right) \textcolor{black}{\left(\begin{bmatrix}
                A_{(s, 1)}B_{s}\\
                \vdots \\
                A_{(s, L)}B_{s} \\
                R_{(s, 1)}B_{s} \\
                \vdots \\
                R_{(s, \ell)}B_{s}
            \end{bmatrix}\right)_{s \in \mathcal{S}_b},}
      \label{Z_def_2} 
\end{align}
 where $\mathcal{I} = [k]$, without loss of generality. The matrix equation in \eqref{Z_def_2} has a unique solution because $M_{\mathcal{I}}^{\mathcal{S}_b}$  defined in \eqref{M_I_Sb} is invertible since its determinant is a minor of a Vandermonde matrix, for which each column $i$ can be factored by $x_{i}^{(b-1)L}$, and $(x_i)_{i \in [N]}$ are non-zero and distinct.    
\end{enumerate}

Finally, from \eqref{Z_i_1} and \eqref{Z_i_2} the total response from Server~$i \in [N]$ is
\begin{align}
    Z_i \overset{\mathrm{\Delta}}{=} (Z_i^{\mathcal{P}}, Z_i^{\overline{\mathcal{P}}})\label{Z_def} .
\end{align}

\subsection{Analysis of the Coding Scheme} \label{ACS}
\subsubsection{Additional Definitions}

\begin{defn}
If $\alpha \geq \frac{\ell}{k}$, then define 
\begin{align}
    \textcolor{black}{\sigma^{\mathcal{S}_b}_{\mathcal{I}}} &\textcolor{black}{\triangleq  \Trans(M_{\mathcal{I}}^{\mathcal{S}_b}) \times A^{\mathcal{S}_b}}\label{A_S_0} \\ 
    &\textcolor{black}{= \Trans(M_{\mathcal{I}}^{\mathcal{S}_b}) \times \left[\Trans(\left[(A_{(s, j)})_{j \in [k]} \right])\right]_{s \in \mathcal{S}_b}}  \nonumber \\
    &\textcolor{black}{= \Trans(M_{\mathcal{I}}^{\mathcal{S}_b}) \times \begin{bmatrix}
A_{((b-1)k+1,1)} & \cdots & A_{(bk,1)} \\
A_{((b-1)k+1,2)}  & \cdots & A_{(bk,2)} \\
\vdots & \ddots & \vdots \\
A_{((b-1)k+1,k)} & \cdots & A_{(bk,k)}
\end{bmatrix}} \nonumber \\ 
&\textcolor{black}{= {\left[\sum_{j \in [k]} A_{(s,j)} x_i^{(b-1)k+j-1}\right]}_{i \in \mathcal{I}, s \in \mathcal{S}_b}} \nonumber \\
    &\textcolor{black}{= {\left[\tilde{A}_{(i,s)}\right]}_{ i \in \mathcal{I}, s \in \mathcal{S}_b},} \nonumber
\end{align}
where $\mathcal{S}_b$ is defined in \eqref{S_b_def_0}, ${A}^{\mathcal{S}_b}$ is defined in \eqref{A_S_b_0},  $M_{\mathcal{I}}^{\mathcal{S}_b}$ is defined in \eqref{M_I_S0}, and $\tilde{A}_{(i,s)}$ is defined in \eqref{A_S_bar_0}.
\end{defn}

\begin{defn}
If $\alpha < \frac{\ell}{k}$, then define 
\begin{align}
    \textcolor{black}{U_b \triangleq \left(\begin{bmatrix}  A_s \\ R_s
         \end{bmatrix}\right)_{s \in \mathcal{S}_b}} ,\label{U_b_mat}
\end{align}
where $b \in \left[\frac{m - p}{L} \right]$, $\mathcal{S}_b$ is defined in \eqref{S_b_def} and $ A^{ \mathcal{S}_b}$ is defined in~\eqref{A_P_bar}. Then, define
\begin{align}
    \textcolor{black}{\sigma^{\mathcal{S}_b}_{\mathcal{I}}} &\textcolor{black}{\triangleq \Trans(M_{\mathcal{I}}^{\mathcal{S}_b}) \times U_b} 
     \label{sigma_Ub} \\
    & \textcolor{black}{= \Trans(M_{\mathcal{I}}^{\mathcal{S}_b}) \times \left[ \begin{array}{c} \Trans([(A_{(s,j)})_{j \in [L]}]) \\  \Trans([(R_{(s, r)})_{r \in [\ell]}]) \end{array} \right]_{s \in \mathcal{S}_b}} \nonumber\\
    &\textcolor{black}{= \Trans(M_{\mathcal{I}}^{\mathcal{S}_b}) \times \begin{bmatrix}
    A_{((b-1)L+1,1)} & \cdots & A_{(bL,1)} \\
    \vdots & \ddots & \vdots \\
    A_{((b-1)L+1,L)} & \cdots & A_{(bL,L)} \\
    R_{((b-1)L+1, 1)} & \cdots & R_{(bL,1)} \\
    \vdots & \ddots & \vdots \\
    R_{((b-1)L+1, \ell)} & \cdots & R_{(bL, \ell)}
    \end{bmatrix}} \nonumber \\
    &\textcolor{black}{= \left[ \sum_{j \in [L]}  x_i^{(b-1)L + j - 1} A_{(s,j)} + \sum_{r \in [\ell]} x_i^{bL + r - 1} R_{(s, r)} \right]_{i \in \mathcal{I}, s \in \mathcal{S}_b}} \nonumber \\
      &\textcolor{black}{= {\left[\tilde{A}_{(i,s)}\right]}_{i \in \mathcal{I}, s \in \mathcal{S}_b},} \label{A_I_P_bar}
\end{align}

where $M_{\mathcal{I}}^{\mathcal{S}_b}$ is defined in \eqref{M_I_Sb} and $\tilde{A}_{(i,s)}$ is defined in \eqref{A_S_bar}. Also, define
\begin{align}
        V\ &\triangleq\ \begin{bmatrix}
			I \\ O
		\end{bmatrix} \in\ \mathbb{F}_{q}^{{k \times L}}, \label{V_def}
\end{align}
where $I \in \mathbb{F}_{q}^{{L \times L}}$ is the identity matrix and $O \in \textcolor{black}{\mathbb{F}_{q}^{(k - L) \times L}}$ has all entries equal to zero such that, for any block $b \in [\frac{m - p}{L}]$, one can recover $A^{ \mathcal{S}_b}$ from $U_b$, as

\begin{align} 
  \textcolor{black}{A^{ \mathcal{S}_b} = \Trans(V) \times U_{b}.} \label{*}
\end{align}
\end{defn}

\subsubsection{Rate}
For the case $\alpha \geq \frac{\ell}{k}$, we first prove the following lemma.
\begin{lem} \label{rate-lm0}
 For the response $Z_i$, $i \in [N]$, defined in \eqref{Z_i_0}, we~have
\begin{align}
    H(Z_i)\ {\leq}\ \left\lceil  \frac{m}{k} \right\rceil{CE}. \label{Z_i_b}
\end{align}
\end{lem}
\begin{proof} 
\begin{align}
    H(Z_i)
    &\overset{\mathrm{(a)}}{=} H\left((Z_i^{\mathcal{S}_b})_{b \in \left[\frac{m}{k}\right]}\right) \nonumber  \displaybreak[0]\\
    &\overset{\mathrm{(b)}}{=} \left\lceil  \frac{m}{k} \right\rceil H(Z_i^{\mathcal{S}_b}) \nonumber  \displaybreak[0]\\
    &\overset{\mathrm{(c)}}{=} \left\lceil  \frac{m}{k} \right\rceil \textcolor{black}{H\left((\tilde{A}_{(i,s)}B_s)_{s \in \mathcal{S}_b}\right)} \nonumber  \displaybreak[0]\\
     &\overset{\mathrm{(d)}}{\leq} \textcolor{black}{\left\lceil  \frac{m}{k}\right\rceil \sum_{s \in \mathcal{S}_b} H\left(\tilde{A}_{(i,s)}B_s\right)} \nonumber  \displaybreak[0]\\
    &\overset{\mathrm{(e)}}{\leq}\left\lceil  \frac{m}{k} \right\rceil{CE}, \nonumber
\end{align}
where
\begin{enumerate}
    \item [(a)] holds by \eqref{Z_i_0};
    \item [(b)] holds by independence of the blocks;
    \item [(c)] holds by \eqref{Z_S_0_1} and  \eqref{Z_S_0};
    \item [(d)] holds by the property of joint entropy;
    \item [(e)] \textcolor{black}{holds because $|\mathcal{S}_b| = k$ and given  \eqref{A_S_bar_0} for any $s \in \mathcal{S}_b$, $\tilde{A}_{(i,s)}B_s  \in \mathbb{F}^{\frac{C}{k} \times E}_{q}$}.
\end{enumerate}    
\end{proof}

Then, we have
\begin{align}
    \Lambda_{k} 
    &=\frac{H(\mathbf{AB}|\mathbf{B})}{\max_{\mathcal{I} \in [N]^{= k}} \sum_{i \in \mathcal{I}}H(Z_i)} \nonumber \displaybreak[0]\\
    &\overset{\mathrm{(a)}}{\geq} \frac{H(\mathbf{AB}|\mathbf{B})}{\left\lceil  \frac{m}{k} \right\rceil kCE} \nonumber  \displaybreak[0]\\
    &\overset{\mathrm{(b)}}{=} \frac{m \cdot \min(\frac{D}{E}, 1)}{k\left\lceil  \frac{m}{k} \right\rceil} \nonumber \displaybreak[0]\\
    &\overset{m \to \infty}{\longrightarrow}  \min\left(\frac{D}{E}, 1\right),\nonumber
\end{align}
 where
\begin{enumerate}
    \item [(a)] holds by \eqref{Z_i_b};
    \item [(b)] holds by \eqref{use_lemma_2_b}.
\end{enumerate}

For the case $\alpha < \frac{\ell}{k}$, from \eqref{Z_def}, we have
\begin{align}
    	&\forall i \in [N],\ H(Z_i) = H(Z_i^{\mathcal{P}}) + H(Z_i^{\overline{\mathcal{P}}}).\label{responses}
\end{align}

Then, we have
\begin{align}
    H(Z_i^{\mathcal{P}})
     &\overset{\mathrm{(a)}}{=} H\left((Z_i^{\mathcal{S}_b})_{b \in \left[\frac{p}{k}\right]}\right) \nonumber  \displaybreak[0]\\
    &\overset{\mathrm{(b)}}{=} {\left\lceil  \frac{p}{k} \right\rceil}H(Z_i^{\mathcal{S}_b}) \nonumber  \displaybreak[0]\\
    &\overset{\mathrm{(c)}}{=} {\left\lceil  \frac{p }{k} \right\rceil}\textcolor{black}{H\left((\tilde{A}_{(i,s)}B_s)_{s \in \mathcal{S}_b}\right)}
 \nonumber  \displaybreak[0]\\
    &\overset{\mathrm{(d)}}{\leq}{\left\lceil  \frac{p }{k} \right\rceil}{CE}, \label{Z_1_upperbound}
\end{align}
where
\begin{enumerate}
    \item [(a)] holds by \eqref{Z_i_1};
    \item [(b)] holds by independence of the blocks;
    \item [(c)] holds by \eqref{Z_S_1};
    \item [(d)] \textcolor{black}{holds because $|\mathcal{S}_b| = k$ and given  \eqref{A_S_bar_0} for any $s \in \mathcal{S}_b$, $\tilde{A}_{(i,s)}B_s  \in \mathbb{F}^{\frac{C}{k} \times E}_{q}$}.
\end{enumerate}
Similarly, we have
\begin{align}
    H(Z_i^{\overline{\mathcal{P}}})\
    {\leq}\ \left\lceil \frac{m - p}{k -\ell}\right\rceil {CE}.\label{Z_2_upperbound} 
\end{align}

Then, from \eqref{responses}, \eqref{Z_1_upperbound}, and \eqref{Z_2_upperbound}, we have
\begin{align}	
    H(Z_i)
    &\leq \left\lceil  \frac{p}{k} \right\rceil CE + \left\lceil \frac{m - p}{k -\ell}\right\rceil CE .\label{Z_i_2_upperbound}
 \end{align} 
 For the communication rate, we have
\begin{align}
    \Lambda_{k} 
    &=\frac{H(\mathbf{AB}|\mathbf{B})}{\max_{\mathcal{I} \in [N]^{= k}} \sum_{i \in \mathcal{I}}H(Z_i)} \nonumber \displaybreak[0]\\
     &\overset{\mathrm{(a)}}{\geq} \frac{H(\mathbf{AB}|\mathbf{B})}{k\left(\left\lceil  \frac{p }{k} \right\rceil + \left\lceil \frac{m - p}{k -\ell}\right\rceil\right) CE} \nonumber	\displaybreak[0]\\
      &\overset{\mathrm{(b)}}{=} 	\frac{m\left(\min(1 , \frac{D}{E})\right)}{k\left(\left\lceil  \frac{p }{k} \right\rceil + \left\lceil \frac{m - p}{k -\ell}\right\rceil\right)} \nonumber \displaybreak[0]\\
       &\overset{m \to \infty}{\longrightarrow} \frac{\min(1 , \frac{D}{E})}{k\left(\frac{1 - \alpha}{k - \ell}\right)},
      \label{rate_lowerbound} 
\end{align}
where 
\begin{enumerate}
    \item [(a)] holds by \eqref{Z_i_2_upperbound};
    \item [(b)] holds by \eqref{use_lemma_2_b};
    \item[]\!\!\!\!\!\!\!\!\!\!\! and in the limit we used \eqref{divide_index}.
\end{enumerate}

\subsubsection{Local Randomness} \label{Randomness}

When $\alpha \geq \frac{\ell}{k}$, we use no randomness in the coding scheme, hence  $R_{(\ell,\alpha, k)} = 0$. However, when $\alpha < \frac{\ell}{k}$, we use randomness \textcolor{black}{in \eqref{A_S_bar}}, where for each of the $\left\lceil  \frac{m - p}{k - \ell} \right\rceil$ blocks, we use a sequence of $\ell$ random matrices that are uniformly distributed over $\textcolor{black}{\mathbb{F}_{q}^{\frac{C}{L} \times D}}$.
\begin{lem} We have
    \begin{align}
        H\left((R^{\mathcal{S}_b})_{b \in \left[  \frac{m - p}{k - \ell} \right]}\right)= \left\lceil  \frac{m - p}{k - \ell}\right\rceil  {\ell CD}. \label{randomness_upperbound}
    \end{align}
\end{lem}
\begin{proof}
We have 
\begin{align}
   H\left((R^{\mathcal{S}_b})_{b \in \left[  \frac{m - p}{k - \ell} \right]}\right) &\overset{\mathrm(a)}{=} \textcolor{black}{\sum_{b \in \left[  \frac{m - p}{k - \ell} \right]} \sum_{s \in \mathcal{S}_b}\sum_{ r \in [\ell]} H\left(R_{(s,r)}\right)} \nonumber  \displaybreak[0]\\
     &\overset{\mathrm(b)}{=} \left\lceil  \frac{m - p}{k - \ell}\right\rceil  {\ell \log|\mathbb{F}_{q}^{ C \times D }|} \nonumber  \displaybreak[0]\\
    &=\left\lceil  \frac{m - p}{k - \ell}\right\rceil  {\ell CD}, \nonumber
\end{align}
where
\begin{enumerate}
    \item [(a)] \textcolor{black}{holds by \eqref{random_mat} and} independence of the blocks and independence of $\ell$ random matrices in each block;
    \item [(b)] \textcolor{black}{holds because $|\mathcal{S}_b|=L$ and for any $s \in \mathcal{S}_b$, $r \in [\ell]$, $R_{(s,r)} \in \mathbb{F}_{q}^{ \frac{C}{L} \times D }$.} 
\end{enumerate}   
\end{proof}
Then, for the rate of local randomness, we have 
\begin{align}
&\frac{H\left((R^{\mathcal{S}_b})_{b \in \left[  \frac{m - p}{k - \ell} \right]}\right)}{H(\mathbf{AB}|\mathbf{B})} \nonumber \displaybreak[0]\\
&\overset{\mathrm(a)}{=}\frac{\left\lceil  \frac{m - p}{k - \ell} \right\rceil \ell CD}{H(\mathbf{AB}|\mathbf{B})} \nonumber \displaybreak[0]\\
    &\overset{\mathrm(b)}{=} \frac{ \left\lceil  \frac{m - p}{k - \ell}\right\rceil \ell CD}{m \cdot \min(CD, CE)} \nonumber \displaybreak[0]\\
    &\overset{\mathrm{(c)}}{=} \frac{\ell  \left\lceil \frac{m - \left\lceil  \alpha \frac{k}{\ell}m \right\rceil}{k -\ell}\right\rceil}{m} \max\left(1, \frac{D}{E}\right)\nonumber  \displaybreak[0]\\
    &\overset{m \to \infty}{\longrightarrow} \frac{\ell - \alpha k}{k - \ell} \max\left(1, \frac{D}{E}\right), \nonumber 
\end{align}
where
\begin{enumerate}
    \item [(a)] holds by \eqref{randomness_upperbound};
    \item [(b)] holds by \eqref{use_lemma_2_b};
    \item [(c)] holds by \eqref{divide_index}.
\end{enumerate}

\subsubsection{Recoverability Constraint}

For each block $S_b$, the user can recover $(A_sB_s)_{s \in S_b}$ from \eqref{Z_def_0} and \eqref{Z_def_2} because the matrices $M_{\mathcal{I}}^{\mathcal{S}_b}$ defined in \eqref{M_I_S0} and \eqref{M_I_Sb} are invertible, as explained in Sections~\ref{CSC_1} and~\ref{CSC_2}, respectively. Thus, the coding scheme designed in Section~\ref{CSC} satisfies~\eqref{recoverability_constraint}. 

\subsubsection{Privacy Constraint}
First, consider the case $\alpha < \frac{\ell}{k}$. 
 \begin{lem} \label{Lemma_3}

 For any $\mathcal{I} \subseteq [N]$, define
\begin{align}
        \textcolor{black}{\tilde{A}_{\mathcal{I}}^{{\mathcal{P}}}} & \textcolor{black}{\overset{\mathrm{\Delta}}{=} ( \tilde{A}_{(i,s)})_{i \in \mathcal{I},\ s \in \mathcal{S}_b,\ b \in \left[\frac{p}{k}\right]},}\nonumber \displaybreak[0]\\
      \textcolor{black}{\tilde{A}_{\mathcal{I}}^{\overline{\mathcal{P}}}} &\textcolor{black}{\overset{\mathrm{\Delta}}{=} (\tilde{A}_{(i,s)})_{i \in \mathcal{I},\ s \in \mathcal{S}_b,\ b \in \left[\frac{m - p}{k - \ell}\right]}.}\label{A_P_bar_I_def}
 \end{align}
Then, we have 
     \begin{align}
         g_1(|\mathcal{I}|)
        &= \frac{H(A^{\mathcal{P}}) - H(A^{\mathcal{P}}|\tilde{A}^{\mathcal{P}}_{\mathcal{I}})}{{(\frac{\alpha}{\ell}k)}^{-1}H(A^{\mathcal{P}})},\label{g1_eq}
    \end{align} \begin{align}
        g_2(|\mathcal{I}|) = \frac{H({A}^{\overline{\mathcal{P}}}) - H({A}^{\overline{\mathcal{P}}}|\tilde{A}^{\overline{\mathcal{P}}}_{\mathcal{I}})}{{\left(1 - \frac{\alpha}{\ell}k \right)}^{-1}H({A}^{\overline{\mathcal{P}}})},\label{g2_eq}
    \end{align}
 where $g_1$ and $g_2$ are defined in \eqref{g1_def} and \eqref{g2_def}, respectively.
 \end{lem}
 
 \begin{proof}
 See Appendix~\ref{appendix_A}. 
 \end{proof}
\begin{lem}\label{Lemma 2}
For any $\mathcal{I} \subseteq [N]$, let $\tilde{\mathbf{A}}_{\mathcal{I}} \triangleq (\tilde{A}_{\mathcal{I}}^{\mathcal{P}}, \tilde{A}^{\overline{\mathcal{P}}}_{\mathcal{I}})$. We have 
    \begin{align}
        &H(\mathbf{A}|\tilde{\mathbf{A}}_{\mathcal{I}}) = \left(1 - g(|\mathcal{I}|)\right)H({\mathbf{A}}).\label{lemma2_result}
    \end{align}
\end{lem}
\begin{proof}
We have 
\begin{align}
    H(\mathbf{A}|\tilde{\mathbf{A}}_{\mathcal{I}})
    &\overset{\mathrm{(a)}}{=} H(A^{\mathcal{P}}|\tilde{A}_{\mathcal{I}}^{\mathcal{P}}) + H({A}^{\overline{\mathcal{P}}}|\tilde{A}^{\overline{\mathcal{P}}}_{\mathcal{I}}) \nonumber  \displaybreak[0]\\
    &\overset{\mathrm{(b)}}{=} \left(1 - {\left(\frac{\alpha}{\ell}k\right)}^{-1}g_1(|\mathcal{I}|)\right)H(A^{\mathcal{P}}) \nonumber\\&+ \left(1 - {\left(1 - \frac{\alpha}{\ell}k \right)}^{-1}g_2(|\mathcal{I}|)\right)H({A}^{\overline{\mathcal{P}}}) \nonumber  \displaybreak[0]\\
  &\overset{\mathrm{(c)}}{=} \left(1 - (g_1(|\mathcal{I}|)+ g_2(|\mathcal{I}|))\right)H(\mathbf{A}) \nonumber  \displaybreak[0]\\
		&\overset{\mathrm{(d)}}{=} \left(1 - g(|\mathcal{I}|)\right)H(\mathbf{A}), \nonumber
\end{align}
 where
    \begin{enumerate}
        \item [(a)] holds by mutual independence between $A^{\mathcal{P}}$ and ${A}^{\overline{\mathcal{P}}}$;
    	\item [(b)] holds by \eqref{g1_eq} and \eqref{g2_eq}; 
    	\item [(c)] holds by  \eqref{divide_index}, \eqref{divide_index_2}, and
     \begin{align}	
        H(A^{\mathcal{P}}) &= \frac{\left\lceil\frac{\alpha}{\ell}km\right\rceil}{m} H(\mathbf{A}) \overset{m \to \infty}{\longrightarrow} \left(\frac{\alpha}{\ell}k\right)H(\mathbf{A}),  \displaybreak[0] \nonumber\\
        H({A}^{\overline{\mathcal{P}}})
        &=  \frac{m- \left\lceil\frac{\alpha}{\ell}km\right\rceil}{m}H(\mathbf{A}) \overset{m \to \infty}{\longrightarrow} \left(1- \frac{\alpha}{\ell}k\right)H(\mathbf{A});\nonumber
    \end{align}
        \item [(d)] holds by \eqref{g_def}.
    \end{enumerate}
\end{proof}
Then, for any $\mathcal{L} \in [N]^{\leq \ell}$, we have 
\begin{align}
    I(\mathbf{A}; \tilde{\mathbf{A}}_{\mathcal{L}})
    &\overset{\mathrm{(a)}}{=} g(|\mathcal{L}|)H(\mathbf{A}) \nonumber \\
    &\overset{\mathrm{(b)}}{\leq} \alpha H(\mathbf{A}),\nonumber
\end{align}
where (a) holds by \eqref{lemma2_result}, and (b) holds by \eqref{g_def}, \eqref{g1_def}, and \eqref{g2_def}.

Then, consider the case $\alpha \geq \frac{\ell}{k}$. 
\begin{lem} \label{Lemma0}
For any $\mathcal{I} \subseteq [N],$ we have
    \begin{align}
        g(|\mathcal{I}|) = \frac{I(\mathbf{A}; \tilde{\mathbf{A}}_{\mathcal{I}})}{H(\mathbf{A})}, \label{g0_eq}
    \end{align} 
    where $g$ is defined in \eqref{case_1_g}. 
\end{lem}
\begin{proof}
    See Appendix~\ref{appendix_B}.
\end{proof}
Then, for any $\mathcal{L} \in [N]^{\leq \ell}$, we have 
\begin{align}
    I(\mathbf{A}; \tilde{\mathbf{A}}_{\mathcal{L}}) 
    &\overset{\mathrm{(a)}}{=} g(|\mathcal{L}|)H(\mathbf{A}) \nonumber \\
    &\overset{\mathrm{(b)}}{\leq} \alpha H(\mathbf{A}), \nonumber
\end{align}
where (a) holds by \eqref{g0_eq}, and (b) holds by \eqref{case_1_g}.

\section{Concluding remarks}\label{Conclusion}
We studied the trade-offs among privacy, communication rate, and local randomness for the problem of secure distributed batch multiplication of matrices that are independently and uniformly distributed over a finite field. \textcolor{black}{By allowing a controlled information leakage, we extended the perfectly secure setting and found new bounds for communication rate and local randomness. Then,} we characterized optimal trade-offs between privacy, communication rate, and optimal rate of local randomness when matrices are square. Finding the capacity for rectangular matrices remains an open problem. 

\appendices \label{appendix}

\section{Proof of Lemma~\ref{LR-lm2}} \label{appendix-1}
To prove Equation \eqref{randomness_converse_2}, we first show that 
\begin{align}
       \sum_{i \in \mathcal{S}}H(\tilde{\mathbf{A}}_{i}|\tilde{\mathbf{A}}_{\mathcal{V}})\ {\geq}\left(1 - \alpha \frac{H(\mathbf{A})}{H(\mathbf{AB}|\mathbf{B})}\right) H(\mathbf{AB}|\mathbf{B}).\label{randomness_converse_1}
    \end{align}
\begin{proof}
We have 
\begin{align}
    \textstyle \sum_{i \in \mathcal{S}}H(\tilde{\mathbf{A}}_{i}|\tilde{\mathbf{A}}_{\mathcal{V}}) &\overset{\mathrm{(a)}}{\geq} H(\tilde{\mathbf{A}}_{\mathcal{S}}|\tilde{\mathbf{A}}_{\mathcal{V}}) \nonumber  \displaybreak[0]\\
    &\overset{\mathrm{(b)}}{\geq} H(\tilde{\mathbf{A}}_{\mathcal{S}}|\tilde{\mathbf{A}}_{\mathcal{W} \cup \mathcal{V}}, \mathbf{B}) \nonumber  \displaybreak[0]\\
    &\geq I(\tilde{\mathbf{A}}_{\mathcal{S}}; \mathbf{A}|\tilde{\mathbf{A}}_{\mathcal{W} \cup \mathcal{V}}, \mathbf{B}) \nonumber  \displaybreak[0]\\
    &= H(\mathbf{A}|\tilde{\mathbf{A}}_{\mathcal{W} \cup \mathcal{V}}, \mathbf{B}) - H(\mathbf{A}|\tilde{\mathbf{A}}_{\mathcal{W} \cup \mathcal{V} \cup \mathcal{S}}, \mathbf{B}) \nonumber  \displaybreak[0]\\
    &\overset{\mathrm{(c)}}{\geq} H(\mathbf{A}|\tilde{\mathbf{A}}_{\mathcal{W} \cup \mathcal{V}}, \mathbf{B}) - H(\mathbf{A}|{Z}_{\mathcal{W} \cup \mathcal{V} \cup \mathcal{S}}, \mathbf{B}) \nonumber  \displaybreak[0]\\
    &\overset{\mathrm{(d)}}{\geq} H(\mathbf{A}|\tilde{\mathbf{A}}_{\mathcal{W} \cup \mathcal{V}}) + H(\mathbf{AB}|\mathbf{B}) - H(\mathbf{A}) \nonumber  \displaybreak[0]\\
    &= H(\mathbf{AB}|\mathbf{B}) - I(\mathbf{A};\tilde{\mathbf{A}}_{\mathcal{W} \cup \mathcal{V}}) \nonumber  \displaybreak[0]\\
    &\overset{\mathrm{(e)}}{\geq} \left(1 - \alpha \frac{H(\mathbf{A})}{H(\mathbf{AB}|\mathbf{B})}\right) H(\mathbf{AB}|\mathbf{B}), \nonumber
    \end{align}
where
\begin{enumerate}
    \item [(a)] holds by the chain rule;
    \item [(b)] holds because conditioning reduces entropy;
    \item [(c)] holds because from \eqref{converse_proof_lm0-1}, we have
    \begin{align}
        I(\mathbf{A}; \tilde{\mathbf{A}}_{\mathcal{W} \cup \mathcal{V} \cup \mathcal{S}},\mathbf{B}) \geq I(\mathbf{A}; Z_{\mathcal{W} \cup \mathcal{V} \cup \mathcal{S}}|\mathbf{B}). \label{converse_proof_c_1}
    \end{align}
    By expanding both mutual information in \eqref{converse_proof_c_1} and using independence between $\mathbf{A}$ and $\mathbf{B}$, we have $H(\mathbf{A}|\tilde{\mathbf{A}}_{\mathcal{W} \cup \mathcal{V} \cup \mathcal{S}}, \mathbf{B}) \leq H(\mathbf{A}|{Z}_{\mathcal{W} \cup \mathcal{V} \cup \mathcal{S}}, \mathbf{B})$; 
   \item [(d)] holds because by \eqref{converse_proof_lm0-2}, we have
    \begin{align}
        I(\mathbf{A}; Z_{\mathcal{W} \cup \mathcal{V} \cup \mathcal{S}}|\mathbf{B}) \geq I(\mathbf{AB}; Z_{\mathcal{W} \cup \mathcal{V} \cup \mathcal{S}}|\mathbf{B}). \label{converse_proof_lm0-1-2}
    \end{align}
    By expanding both sides of \eqref{converse_proof_lm0-1-2}, we have
    \begin{align}
        &H(\mathbf{A}|\mathbf{B}) - H(\mathbf{A}|{Z}_{\mathcal{W} \cup \mathcal{V} \cup \mathcal{S}}, \mathbf{B}) \nonumber  \displaybreak[0]\\
        &\geq H(\mathbf{AB}|\mathbf{B}) - H(\mathbf{AB}|{Z}_{\mathcal{W} \cup \mathcal{V} \cup \mathcal{S}}, \mathbf{B}) \nonumber  \displaybreak[0]\\
        &= H(\mathbf{AB}|\mathbf{B}), \label{converse_e}
     \end{align}
    where \eqref{converse_e} holds by \eqref{recoverability_constraint} and $|\mathcal{W} \cup \mathcal{V} \cup \mathcal{S}| = k$. Then,  from \eqref{converse_e} and independence of $\mathbf{A}$ and $\mathbf{B}$, we have
    \begin{align}
        H(\mathbf{A}|\tilde{\mathbf{A}}_{\mathcal{W} \cup \mathcal{V}}, \mathbf{B}) - H(\mathbf{A}|{Z}_{\mathcal{W} \cup \mathcal{V} \cup \mathcal{S}}, \mathbf{B}) \geq H(\mathbf{A}|\tilde{\mathbf{A}}_{\mathcal{W} \cup \mathcal{V}}) + H(\mathbf{AB}|\mathbf{B}) - H(\mathbf{A});\nonumber
    \end{align}
    \item [(e)] holds by \eqref{privacy_constraint} because $|\mathcal{W} \cup \mathcal{V}| = \ell$. 
\end{enumerate} 
\end{proof}
Then, with $T \triangleq \frac{1}{k - \ell}\binom{N - v}{\ell - v}^{-1}\binom{N - \ell}{k - \ell}^{-1}$, we have
    \begin{align}
        &\frac{1 -\alpha \frac{H(\mathbf{A})}{H(\mathbf{AB}|\mathbf{B})}}{k - \ell}H(\mathbf{AB}|\mathbf{B}) \displaybreak[0] \nonumber \displaybreak[0]\\
        &= T\sum_{\substack{\mathcal{W} \subseteq [N] \setminus \mathcal{V}\ \nonumber \\ |\mathcal{W}| = \ell - v}}\sum_{\substack{\mathcal{S} \subseteq [N] \setminus \mathcal{W} \cup \mathcal{V}\ \nonumber \\ |\mathcal{S}| = k - \ell}} \left(1 - \alpha \frac{H(\mathbf{A})}{H(\mathbf{AB}|\mathbf{B})} \right){H(\mathbf{AB}|\mathbf{B})} \nonumber \displaybreak[0]\\
        &\overset{(a)}{\leq}  T\sum_{\substack{\mathcal{W} \subseteq [N] \setminus \mathcal{V}\ \nonumber \\ |\mathcal{W}| = \ell - v}}\sum_{\substack{\mathcal{S} \subseteq [N] \setminus (\mathcal{W} \cup \mathcal{V})\ \nonumber \\ |\mathcal{S}| = k - \ell}} \sum_{i \in \mathcal{S}}H(\tilde{\mathbf{A}}_{i}|\tilde{\mathbf{A}}_{\mathcal{V}}) \nonumber \displaybreak[0]\\
        &\overset{(b)}{=} T\sum_{\substack{\mathcal{W} \subseteq [N] \setminus \mathcal{V}\ \nonumber \\ |\mathcal{W}| = \ell - v}} \binom{N - \ell - 1}{k - \ell - 1} \sum_{\substack{i \in [N] \setminus (\mathcal{W} \cup \mathcal{V})}} H(\tilde{\mathbf{A}}_{i}|\tilde{\mathbf{A}}_{\mathcal{V}}) \nonumber \displaybreak[0]\\
        &\overset{(c)}{=} T\binom{N - \ell - 1}{k - \ell - 1}\sum_{\substack{\mathcal{W} \subseteq [N] \setminus \mathcal{V}\ \nonumber \\ |\mathcal{W}| = N - \ell}}\sum_{\substack{i \in \mathcal{W}}} H(\tilde{\mathbf{A}}_{i}|\tilde{\mathbf{A}}_{\mathcal{V}}) \nonumber \displaybreak[0]\\
        &\overset{(d)}{=} T\binom{N - \ell - 1}{k - \ell - 1}\binom{N - v - 1}{N - \ell - 1}\sum_{\substack{i \in [N] \setminus \mathcal{V}}} H(\tilde{\mathbf{A}}_{i}|\tilde{\mathbf{A}}_{\mathcal{V}}) \nonumber \displaybreak[0]\\
        &=\frac{1}{N - v}\sum_{i \in [N] \setminus \mathcal{V}}H(\tilde{\mathbf{A}}_{i}|\tilde{\mathbf{A}}_{\mathcal{V}}) \nonumber \displaybreak[0]\\
        &\overset{(e)}{\leq} \frac{1}{N - v}\sum_{i \in [N] \setminus \mathcal{V}}H(\tilde{\mathbf{A}}_{i^{*}(\mathcal{V})}|\tilde{\mathbf{A}}_{\mathcal{V}}) \nonumber \displaybreak[0]\\
        &= H(\tilde{\mathbf{A}}_{i^{*}(\mathcal{V})}|\tilde{\mathbf{A}}_{\mathcal{V}}) \nonumber \displaybreak[0]\\
        &=H(\tilde{\mathbf{A}}_{[N]}|\tilde{\mathbf{A}}_{\mathcal{V}}) - H(\tilde{\mathbf{A}}_{[N]}|\tilde{\mathbf{A}}_{\mathcal{V} \cup \{i ^{*} (\mathcal{V})\}}), \nonumber
    \end{align}
where
    \begin{enumerate}
        \item [(a)] holds by \eqref{randomness_converse_1};
        \item [(b)] holds because for any $i \in [N] \setminus (\mathcal{W} \cup \mathcal{V})$, the term $H(\tilde{\mathbf{A}}_{i}|\tilde{\mathbf{A}}_{\mathcal{V}})$ appears exactly $\binom{N - \ell - 1}{k - \ell - 1}$ times in the term $\sum_{\substack{\mathcal{S} \subseteq [N] \setminus (\mathcal{W} \cup \mathcal{V})\ \\ |\mathcal{S}| = k - \ell}} \sum_{i \in \mathcal{S}}H(\tilde{\mathbf{A}}_{i}|\tilde{\mathbf{A}}_{\mathcal{V}})$, this argument is similar to \cite[Lemma 3.2]{de1999multiple};
        \item [(c)] holds by a change of variables in the sums;
        \item [(d)] holds because for any $i \in [N] \setminus \mathcal{V}$, $H(\tilde{\mathbf{A}}_{i}|\tilde{\mathbf{A}}_{\mathcal{V}})$ appears exactly $\binom{N - v - 1}{N - \ell - 1}$ times in the term $\sum_{\substack{\mathcal{W} \subseteq [N] \setminus \mathcal{V}\ \\ |\mathcal{W}| = N - \ell}}\sum_{\substack{i \in \mathcal{W}}} H(\tilde{\mathbf{A}}_{i}|\tilde{\mathbf{A}}_{\mathcal{V}})$;
        \item [(e)] holds with $i^{*}(\mathcal{V}) \in \arg\max_{i \in [N] \setminus \mathcal{V}} H(\tilde{\mathbf{A}}_{i}|\tilde{\mathbf{A}}_{\mathcal{V}})$. 
    \end{enumerate}

\section{Proof of Lemma~\ref{Lemma_3}} \label{appendix_A}

Here for any $\mathcal{I} \subseteq [N]$, we prove \eqref{g2_eq}. Given \eqref{g2_def} and $i \triangleq |\mathcal{I}|$, we need to show 
\begin{align}
   & \frac{H({A}^{\overline{\mathcal{P}}}) - H({A}^{\overline{\mathcal{P}}}|\tilde{A}^{\overline{\mathcal{P}}}_{\mathcal{I}})}{H({A}^{\overline{\mathcal{P}}})} \nonumber \\ &= \left\{\begin{matrix}
    & 0,\ & i \in [0: \ell]\\
    		& \frac{ \frac{1 - \alpha}{k -\ell} (i - \ell) + \alpha -   \left(\frac{\alpha}{\ell} i\right) }{1 -  \frac{\alpha}{\ell} k} ,\ & i \in (\ell: k)\\
    		& 1 & i \in [k: N]
    	\end{matrix}\right. \ , \nonumber
\end{align}
i.e.,
\begin{align}
    H({A}^{\overline{\mathcal{P}}}|\tilde{A}^{\overline{\mathcal{P}}}_{\mathcal{I}}) = \left\{\begin{matrix}
    & H({A}^{\overline{\mathcal{P}}}),\ & i \in [0: \ell]\\
    		& \frac{k - i}{k - \ell}H({A}^{\overline{\mathcal{P}}}),\ & i \in (\ell: k)\\
    		& 0 & i \in [k: N]
    	\end{matrix}\right. \ .\label{g_2_to_prove}
\end{align}

Note that to prove \eqref{g1_eq}, given \eqref{g1_def}, we need to show

\begin{align}
   & \frac{H({A}^{{\mathcal{P}}}) - H({A}^{{\mathcal{P}}}|\tilde{A}^{{\mathcal{P}}}_{\mathcal{I}})}{H({A}^{{\mathcal{P}}})} = \left\{\begin{matrix}
    & \frac{i}{k},\ & i \in [0: k)\\
    		& 1, \ & i \in [k: N]
    	\end{matrix}\right. \ , \nonumber
\end{align}
i.e.,
\begin{align}
   H({A}^{{\mathcal{P}}}|\tilde{A}^{{\mathcal{P}}}_{\mathcal{I}}) = \left\{
    \begin{matrix}
        & \left(\frac{k - i}{k}\right){H({A}^{{\mathcal{P}}})},\ & i \in [0: k)\\
    	& 0 & i \in [k: N]
    \end{matrix}\right. \ ,\label{g_0_to_prove}
\end{align}
which is a special case of \eqref{g_2_to_prove} with the substitutions ${{\mathcal{P}}} \leftarrow \overline{\mathcal{P}}$ and $\ell \leftarrow 0$. 

Let $L \triangleq k - \ell$  and $t \triangleq {[k - i]}^{+}$.  We rewrite  \eqref{g_2_to_prove} as follows 
\begin{align}
    t = 0 \Rightarrow H({A}^{\overline{\mathcal{P}}}|\tilde{A}^{\overline{\mathcal{P}}}_{\mathcal{I}}) &= 0, \label{t_0} \\
    t \in [1: L) \Rightarrow H({A}^{\overline{\mathcal{P}}}|\tilde{A}^{\overline{\mathcal{P}}}_{\mathcal{I}}) &=  \frac{t}{L}H({A}^{\overline{\mathcal{P}}}), \label{t_1_L} \\
    t \in [L:k] \Rightarrow H({A}^{\overline{\mathcal{P}}}|\tilde{A}^{\overline{\mathcal{P}}}_{\mathcal{I}}) &=  H({A}^{\overline{\mathcal{P}}}). \label{t_L_k}
\end{align}

When $t=0$, we have $|\mathcal{I}| \geq k$, without loss of generality consider one block $ \mathcal{S}_b$, $b \in [J]$, and  ${\overline{{\mathcal{P}}}} \leftarrow \mathcal{S}_b$ as the blocks are independent. \textcolor{black}{Given $\tilde{A}_{\mathcal{I}}^{\mathcal{S}_b}$, we can compute $\sigma^{\mathcal{S}_b}_{\mathcal{I}}$ from \eqref{A_I_P_bar}. The matrix equation in \eqref{sigma_Ub} has a unique solution since $ M_{\mathcal{I}}^{\mathcal{S}_b}$ defined in \eqref{M_I_Sb} is invertible, and with $U_b$, one can determine ${A}^{\mathcal{S}_b}$ from \eqref{*}}; hence, we deduce $H({A}^{\overline{\mathcal{P}}}|\tilde{A}^{\overline{\mathcal{P}}}_{\mathcal{I}}) = 0$. \textcolor{black}{ Similarly, for \eqref{g_0_to_prove} when $i \geq k$, we have $H({A}^{{\mathcal{P}}}|\tilde{A}^{{\mathcal{P}}}_{\mathcal{I}}) = 0$ as the matrix equation in \eqref{A_S_0} has a unique solution. }

For $t \in [L: k]$, we have $\mathcal{|I|} \leq \ell$. For any $\mathcal{I} \in [N]^{\leq \ell}$, we prove  \eqref{t_L_k} in Appendix \ref{Appendix_A_3_Proof}. 

Finally, for $t \in [1 : L)$, we have $\mathcal{|I|} = k - t$.  For any $ \mathcal{I} \in [N]^{= k - t}$, we prove \eqref{t_1_L} in Appendix \ref{Appendix_A_Proof}.  

\section{Proof of \eqref{t_L_k}} \label{Appendix_A_3_Proof}
With $\mathcal{I} \in [N]^{\leq \ell}$, we provide the proof for one block $S_b$, for any $b \in [J]$, without loss of generality as the blocks are independent. We have 
\begin{align}
H({A}^{{\mathcal{S}_b}}|\tilde{A}^{{\mathcal{S}_b}}_{\mathcal{I}}) \leq H({A}^{{\mathcal{S}_b}}), \label{need_to_prove_1_1}
\end{align}
since conditioning reduces entropy. Also, we have 
\begin{align}
     &H({A}^{{\mathcal{S}_b}}|\tilde{A}^{{\mathcal{S}_b}}_{\mathcal{I}}) \nonumber\\
     &= H({A}^{{\mathcal{S}_b}},\tilde{A}^{{\mathcal{S}_b}}_{\mathcal{I}}) - H(\tilde{A}^{{\mathcal{S}_b}}_{\mathcal{I}}) \nonumber  \displaybreak[0]\\
     &\overset{\mathrm{(a)}}{=}  
      H\left({A}^{{\mathcal{S}_b}}, \textcolor{black}{\sigma^{\mathcal{S}_b}_{\mathcal{I}}}\right) - H(\textcolor{black}{\tilde{A}^{\mathcal{S}_{b}}_{\mathcal{I}}}) \nonumber  \displaybreak[0]\\
     &\overset{\mathrm{(b)}}{=} H\left({A}^{{\mathcal{S}_b}}, R^{\mathcal{S}_{b}}\right) - H(\textcolor{black}{\tilde{A}^{\mathcal{S}_{b}}_{\mathcal{I}}}) \nonumber  \displaybreak[0]\\
     &=H({A}^{{\mathcal{S}_b}}) +  \textcolor{black}{\sum_{s \in \mathcal{S}_b}\sum_{r \in [\ell]} H(R_{(s,r)})} - H(\textcolor{black}{\tilde{A}^{\mathcal{S}_{b}}_{\mathcal{I}}})\nonumber \displaybreak[0]\\
     &\overset{\mathrm{(c)}}{=} \log|\left(\mathbb{F}_q^{C \times D}\right)^{1 \times L}| +  \ell \log|\mathbb{F}_q^{C \times D}| - H(\textcolor{black}{\tilde{A}^{\mathcal{S}_{b}}_{\mathcal{I}}})\nonumber  \displaybreak[0]\\
     &\overset{\mathrm{(d)}}{\geq}  k\log|\mathbb{F}_q^{C \times D}| -  |\mathcal{I}|\log|\mathbb{F}_q^{C \times D}|\nonumber  \displaybreak[0]\\
     &\overset{\mathrm{(e)}}{\geq}  k\log|\mathbb{F}_q^{C \times D}| -  \ell\log|\mathbb{F}_q^{C \times D}|\nonumber  \displaybreak[0]\\
      &=  L\log|\mathbb{F}_q^{C \times D}| \nonumber  \displaybreak[0]\\
      &= H({A}^{{\mathcal{S}_b}}), \label{need_to_prove_1_2}
\end{align}
where
\begin{enumerate}
    \item [(a)] \textcolor{black}{holds by \eqref{A_I_P_bar}};
    \item [(b)] \textcolor{black}{holds by \eqref{U_b_mat}, \eqref{sigma_Ub}};
    \item [(c)] holds by \eqref{A_P_bar} where ${A}^{\mathcal{S}_b} \in \left(\mathbb{F}_q^{C \times D}\right)^{1 \times L}$ \textcolor{black}{and \eqref{random_mat} where we use a sequence of $\ell$ independent matrices $R_{(s,r)} \in \mathbb{F}_q^{\frac{C}{L} \times D}$};
    \item [(d)] holds by $L=k-\ell$ and \textcolor{black}{given  \eqref{A_S_bar} and \eqref{A_P_bar_I_def},} we have $\textcolor{black}{H(\tilde{A}^{\mathcal{S}_{b}}_{\mathcal{I}}) \leq |\mathcal{S}_b|\times |\mathcal{I}| \log\left|\mathbb{F}_{q}^{\frac{C }{L}\times D}\right|= |\mathcal{I}|\log\left|\mathbb{F}_{q}^{C \times D}\right|}$;
    \item [(e)] holds because $|\mathcal{I}| \leq \ell$.
\end{enumerate}
From \eqref{need_to_prove_1_1} and \eqref{need_to_prove_1_2}, we deduce \eqref{t_L_k}.

\section{Proof of \eqref{t_1_L}} \label{Appendix_A_Proof}

We will prove  \eqref{appendix_lemma_1} and \eqref{appendix_lemma_2}  for $\mathcal{I} \in [N]^{= k - t}$
    \begin{align}    &H({A}^{\overline{\mathcal{P}}}|\tilde{A}^{\overline{\mathcal{P}}}_{\mathcal{I}}) \leq \frac{t}{L}H({A}^{\overline{\mathcal{P}}}) \label{appendix_lemma_1},  \displaybreak[0]\\
    &H({A}^{\overline{\mathcal{P}}}|\tilde{A}^{\overline{\mathcal{P}}}_{\mathcal{I}}) \geq \frac{t}{L}H({A}^{\overline{\mathcal{P}}}).\label{appendix_lemma_2}
    \end{align}

Without loss of generality, consider one block  ($\overline{\mathcal{P}} = \mathcal{S}_b$, $b \in [J]$), as the blocks are independent. We prove \eqref{appendix_lemma_1} by induction in Appendix~\ref{sub_appendix_A_1}. We prove \eqref{appendix_lemma_2} in Appendix~\ref{sub_appendix_A_2}.  

\subsection{Proof of \eqref{appendix_lemma_1}} \label{sub_appendix_A_1}
We employ a proof by induction on $t$. 
For $t = 0$, we have $H({A}^{\mathcal{S}_b}|\tilde{A}^{\mathcal{S}_b}_{\mathcal{I}}) = 0$ because  $|\mathcal{I}| = k$, and \textcolor{black}{the matrix equation in \eqref{sigma_Ub} has a unique solution since $ M_{\mathcal{I}}^{\mathcal{S}_b}$ defined in \eqref{M_I_Sb} is invertible}.  For any $\mathcal{I} \in [N]^{= k - t}$, suppose that \eqref{appendix_lemma_1} holds. Then, consider $t^{'} \triangleq t+1$ and $\mathcal{I'} \in [N]^{=k - t^{'}}$. 

For any $\mathcal{I}$ that contains $\mathcal{I}^{'}$, we have 
\begin{align}
    H({A}^{\mathcal{S}_b}|\tilde{A}^{\mathcal{S}_b}_{\mathcal{I}}) &= H({A}^{\mathcal{S}_b}, \tilde{A}^{\mathcal{S}_b}_{\mathcal{I}}) - H(\tilde{A}^{\mathcal{S}_b}_{\mathcal{I}}) \nonumber  \displaybreak[0]\\
    &\overset{\mathrm{(a)}}{\geq} H({A}^{\mathcal{S}_b}, \tilde{A}^{\mathcal{S}_b}_{\mathcal{I'}}) \nonumber  \displaybreak[0]\\ &- \left(H(\tilde{A}^{\mathcal{S}_b}_{\mathcal{I'}}) + H(\tilde{A}^{\mathcal{S}_b}_{\mathcal{I} \setminus \mathcal{I'}}\ |\tilde{A}^{\mathcal{S}_b}_{\mathcal{I'}})\right) \nonumber  \displaybreak[0]\\
     &\overset{\mathrm{(b)}}{=} H({A}^{\mathcal{S}_b} | \tilde{A}^{\mathcal{S}_b}_{\mathcal{I'}}) - H(\tilde{A}^{ \mathcal{S}_b}_{\mathcal{I} \setminus \mathcal{I'}} |\tilde{A}^{\mathcal{S}_b}_{\mathcal{I'}})  \nonumber  \displaybreak[0]\\
    &\overset{\mathrm{(c)}}{\geq} H({A}^{\mathcal{S}_b} |\tilde{A}^{\mathcal{S}_b}_{\mathcal{I'}}) - H(\tilde{A}^{\mathcal{S}_b}_{\mathcal{I} \setminus \mathcal{I'}}) \nonumber  \displaybreak[0]\\
    &\textcolor{black}{\overset{\mathrm{(d)}}{\geq} H({A}^{\mathcal{S}_b} |\tilde{A}^{\mathcal{S}_b}_{\mathcal{I'}}) - \log|\mathbb{F}_q{^{C \times D}}|} \nonumber  \displaybreak[0]\\
	&\overset{\mathrm{(e)}}{=} H({A}^{\mathcal{S}_b} |\tilde{A}^{\mathcal{S}_b}_{\mathcal{I'}}) -\frac{1}{L}H({A}^{\mathcal{S}_b}),  \nonumber  \displaybreak[0]
\end{align}
where
\begin{enumerate}
    \item [(a)] and (b) hold by the chain rule;
    \item [(c)] holds because conditioning reduces entropy;
    \item [(d)] holds because \textcolor{black}{given \eqref{A_S_bar} and \eqref{A_P_bar_I_def},} we have \textcolor{black}{ $\tilde{A}^{\mathcal{S}_b}_{\mathcal{I} \setminus \mathcal{I'}} = (\tilde{A}_{(i,s)})_{s \in \mathcal{S}_b, i \in \mathcal{I} \setminus \mathcal{I'}} \in\ \mathbb{F}_{q}^{C \times D}$ because $|\mathcal{I} \setminus \mathcal{I'}| = 1$}; 
    \item [(e)] holds by \eqref{A_P_bar} where ${A}^{\mathcal{S}_b} \in \left(\mathbb{F}_q{^{{C \times D}}}\right)^{1 \times L}$.
\end{enumerate}
Consequently, we have
\begin{align}
H({A}^{\mathcal{S}_b}|\tilde{A}^{\mathcal{S}_b}_{\mathcal{I'}}) &\leq H({A}^{\mathcal{S}_b}|\tilde{A}^{\mathcal{S}_b}_{\mathcal{I}}) +\frac{1}{L}H({A}^{\mathcal{S}_b}) \nonumber \\  
&\leq \frac{t+1}{L}H({A}^{\mathcal{S}_b} ), \nonumber
\end{align}
where the last inequality holds by the induction hypothesis. 

\subsection{Proof of \eqref{appendix_lemma_2}} \label{sub_appendix_A_2}

\begin{lem} \label{G_I_lem}
Let $\mathcal{I} \in [N]^{= k - L}$, $b \in [J]$, we have
  \begin{align}
      H(U_b|A^{\mathcal{S}_b}, \tilde{A}^{\mathcal{S}_b}_{\mathcal{I}}) = 0. \label{HUSB_0}
  \end{align}
\end{lem}
\begin{proof}
    See Appendix~\ref{HU_0}.
\end{proof} 
For any $\mathcal{I} \in [N]^{= k - t}$, we have 
\begin{align} 
    &H(A^{\mathcal{S}_b}|\tilde{A}^{\mathcal{S}_b}_{\mathcal{I}}) \nonumber  \displaybreak[0]\\
    &{\geq} I(A^{\mathcal{S}_b}; U_b|\tilde{A}^{\mathcal{S}_b}_{\mathcal{I}}) \nonumber  \displaybreak[0]\\
    &= H(U_b|\tilde{A}^{\mathcal{S}_b}_{\mathcal{I}}) - H(U_b|A^{\mathcal{S}_b}, \tilde{A}^{\mathcal{S}_b}_{\mathcal{I}}) \nonumber  \displaybreak[0]\\
    &\overset{\mathrm{(a)}}{=} H(U_b|\tilde{A}^{\mathcal{S}_b}_{\mathcal{I}}) \nonumber  \displaybreak[0]\\
    &\overset{\mathrm{(b)}}{\geq} H(U_b) - H(\tilde{A}^{\mathcal{S}_b}_{\mathcal{I}}) \nonumber  \displaybreak[0]\\
    &\overset{\mathrm{(c)}}{\geq} H(U_b) - ({k - t})\log|\mathbb{F}_{q}^{C \times D}| \nonumber \displaybreak[0] \\
    &\overset{\mathrm{(d)}}{=} H(A^{\mathcal{S}_b}) + (t - L)\log|\mathbb{F}_q^{C \times D}| \nonumber \displaybreak[0] \\
    &\overset{\mathrm{(e)}}{=} \frac{t}{L}H(A^{\mathcal{S}_b}), \label{last_step}
\end{align}
where
\begin{enumerate}
    \item [(a)] holds by Lemma~\ref{G_I_lem}; 
    \item [(b)] holds because 
\begin{align}
    H(U_{b}|\tilde{A}^{\mathcal{S}_b}_{\mathcal{I}})  &= H(U_{b}, \tilde{A}^{\mathcal{S}_b}_{\mathcal{I}}) - H(\tilde{A}^{\mathcal{S}_b}_{\mathcal{I}})  \nonumber \\ &\geq H(U_{b}) -H(\tilde{A}^{\mathcal{S}_b}_{\mathcal{I}}); \nonumber
\end{align}
    \item [(c)] holds because $|\mathcal{I}| = k-t$ and \textcolor{black}{given \eqref{A_S_bar} and \eqref{A_P_bar_I_def},}  we have \textcolor{black}{$
    H(\tilde{A}^{\mathcal{S}_b}_{\mathcal{I}})\leq |\mathcal{I}| \times |\mathcal{S}_b| \log\left|\mathbb{F}_{q}^{\frac{C }{L}\times D}\right|= (k-t)\log\left|\mathbb{F}_{q}^{C \times D}\right|$;}
    \item [(d)] holds because 
    \begin{align} 
        H(U_{b})   
        &= H(A^{\mathcal{S}_b}) + H(R^{ \mathcal{S}_b}) \nonumber  \displaybreak[0]\\
        &= H(A^{\mathcal{S}_b}) + (k - L)\log|\mathbb{F}_{q}^{C \times D}|; \label{**}
    \end{align}
    \item [(e)] holds by \eqref{A_P_bar} where $A^{\mathcal{S}_b}$ is uniformly distributed over ${\left(\mathbb{F}_{q}^{C \times D}\right)}^{1 \times L}$. 
\end{enumerate}

\section{Proof of Lemma~\ref{G_I_lem}} \label{HU_0}
For any $ \mathcal{I} \in [N]^{= k - L}$, $b \in [J]$, define
\begin{align}
G^{\mathcal{S}_b}_{\mathcal{I}} &\triangleq\  \begin{bmatrix}
    \Trans(V) \\ \textcolor{black}{\Trans(M_{\mathcal{I}}^{\mathcal{S}_b})}\\
    \end{bmatrix}  \in\ \mathbb{F}_{q}^{k \times k},  \label{G_b_def}
\end{align}
where $V$ and $M_{\mathcal{I}}^{\mathcal{S}_b}$ are defined in \eqref{V_def} and \eqref{M_I_Sb}, respectively. Then, from \eqref{A_I_P_bar}, \eqref{*}, and \eqref{G_b_def}, we have
\begin{align}
    \textcolor{black}{G^{\mathcal{S}_b}_{\mathcal{I}} \times U_b} &= \begin{bmatrix}
    I \ \ \ \ \ \ O \\ \textcolor{black}{\Trans(M_{\mathcal{I}}^{\mathcal{S}_b})}\\
    \end{bmatrix} \times  \textcolor{black}{\left[ \begin{array}{c} \Trans([(A_{(s,j)})_{j \in [L]}]) \\  \Trans([(R_{(s, r)})_{r \in [\ell]}]) \end{array} \right]_{s \in \mathcal{S}_b}} = \begin{bmatrix} A^{{\mathcal{S}_b}} \\ \sigma^{{\mathcal{S}_b}}_{\mathcal{I}}\end{bmatrix}. \label{UG_2}
\end{align}
We now prove that $G^{\mathcal{S}_b}_{\mathcal{I}}$ is invertible. We analyze the first block $b=1$ without loss of generality since the blocks are independent. Let $\mathcal{I} = [k-L]$, we have 
\begin{align}
     G^{\mathcal{S}_1}_{\mathcal{I}} &= \begin{bmatrix} \Trans(V) \\ \textcolor{black}{\Trans(M_{\mathcal{I}}^{\mathcal{S}_1})}\end{bmatrix} = \begin{bmatrix}
I & O\\
\textcolor{black}{\Trans}(M^{'}) & \textcolor{black}{\Trans}( M^{''})
\end{bmatrix} \in\ \mathbb{F}_{q}^{k \times k},\nonumber
\end{align} where \( I \in \mathbb{F}_{q}^{L \times L} \) is the identity matrix,  \( O \in \mathbb{F}_{q}^{(k - L) \times L} \) is the zero matrix, and
\begin{align}
M^{'} &\triangleq \begin{bmatrix}
1 & 1 & \ldots & 1 \\
x_1 & x_2 & \ldots & x_{k-L} \\
\vdots & \vdots & \ldots & \vdots \\
x_1^{L-1} & x_2^{L-1} & \ldots & x_{k-L}^{L-1}
\end{bmatrix} \in {\mathbb{F}_{q}}^{L \times \ ( k-L)},\nonumber\\ \displaybreak[0]  \nonumber  \displaybreak[0]\\
M^{''} &\triangleq \begin{bmatrix}
x_1^L & x_2^L & \ldots & x_{k-L}^L \\
x_1^{L+1} & x_2^{L+1} & \ldots & x_{k-L}^{L+1} \\
\vdots & \vdots & \ddots & \vdots \\
x_1^{k-1} & x_2^{k-1} & \ldots & x_{k-L}^{k-1}
\end{bmatrix} \in {\mathbb{F}_{q}}^{( k-L) \times \ ( k-L)}.\nonumber
\end{align}
Then, we have
\begin{align}
	det (G^{\mathcal{S}_1}_\mathcal{I})  \nonumber  \displaybreak[0] &= \ det\left(\Trans({M^{''}})\right) .  \nonumber  
\end{align}
Similar to the discussion regarding the matrix in \eqref{Z_def_2}, $\Trans({M^{''}})$ is invertible, and thus 
$G^{\mathcal{S}_1}_\mathcal{I}$ is invertible.

\section{Proof of Lemma~\ref{Lemma0}} \label{appendix_B}
To show  \eqref{g0_eq}, it is sufficient,  given \eqref{case_1_g}, to show that, for any $\mathcal{I} \subseteq [N]$ with $i \triangleq |\mathcal{I}|$, we have
\begin{align}
    \frac{H(\mathbf{A}) - H(\mathbf{A}|\tilde{\mathbf{A}}_{\mathcal{I}})}{H(\mathbf{A})} = \left\{\begin{matrix}
    & \frac{i}{k} & i \in [0: k)\\
    & 1 & i \in [k: N]
    	\end{matrix}\right. \ , \nonumber
\end{align}
i.e.,
\begin{align}
    H(\mathbf{A}|\tilde{\mathbf{A}}_{\mathcal{I}}) = \left\{
    \begin{matrix}
        & \left(\frac{k - i}{k}\right)H(\mathbf{A}),\ & i \in [0: k)\\
    	& 0 & i \in [k: N]
    \end{matrix}\right. \ .\label{g_0_0_to_prove}
\end{align}
Let $t \triangleq {[k - i]}^{+}$ and rewrite \eqref{g_0_0_to_prove} as
\begin{align}
    t = 0 \Rightarrow H(\mathbf{A}|\tilde{\mathbf{A}}_{\mathcal{I}}) &= 0, \label{t_0_0}  \displaybreak[0]\\
    t \in [1: k) \Rightarrow H(\mathbf{A}|\tilde{\mathbf{A}}_{\mathcal{I}}) &= \frac{t}{k}H(\mathbf{A}), \label{t_1_0} \displaybreak[0]\\
    t = k \Rightarrow H(\mathbf{A}|\tilde{\mathbf{A}}_{\mathcal{I}}) &= H(\mathbf{A}). \label{t_2_0}
\end{align}

The expressions in \eqref{t_0_0}, \eqref{t_1_0}, and \eqref{t_2_0} are a special case of \eqref{t_0}, \eqref{t_1_L}, and \eqref{t_L_k}, respectively, with the substitutions $\overline{\mathcal{P}} \leftarrow [m]$ and $\ell \leftarrow 0$.   

\bibliography{refs}
\bibliographystyle{ieeetr}

\end{document}